\def\version{PoPETS}
\newif \ifsubmission \submissiontrue

\newif \iffull 
\newif \ifACM
\newif \ifUSENIX
\newif \ifIEEE
\newif \ifLNCS
\newif \ifPoPETS

\newif \ifCCS
\newif \ifSP
\newif \ifNDSS
\newif \ifCrypto
\newif \ifFC

\def\fullstring{full}
\def\ACMstring{ACM}
\def\USENIXstring{USENIX}
\def\IEEEstring{IEEE}
\def\LNCSstring{LNCS}
\def\PoPETSstring{PoPETS}
\def\CCSstring{CCS}
\def\SPstring{SP}
\def\NDSSstring{NDSS}
\def\Cryptostring{CRYPTO}
\def\FCstring{FC}

\ifx \version\fullstring \fulltrue \fi
\ifx \version\ACMstring \ACMtrue \fi
\ifx \version\USENIXstring \USENIXtrue \fi
\ifx \version\IEEEstring \IEEEtrue \fi
\ifx \version\LNCSstring \LNCStrue \fi
\ifx \version\PoPETSstring \PoPETStrue \fi
\ifx \version\CCSstring \ACMtrue \CCStrue \fi
\ifx \version\SPstring \IEEEtrue \SPtrue \fi
\ifx \version\NDSSstring \IEEEtrue \NDSStrue \fi
\ifx \version\Cryptostring \LNCStrue \Cryptotrue \fi
\ifx \version\FCstring \LNCStrue \FCtrue \fi

\ifPoPETS \documentclass[sigconf,nonacm]{acmart}

\AtBeginDocument{%
  \fancypagestyle{standardpagestyle}{%
    \fancyhf{}%
    \fancyfoot[C]{\thepage}%
  }%
  \pagestyle{standardpagestyle}%
}
 \fi

\iffull \bibliography{references} \fi
\newif \ifcomments \commentsfalse
\newif \ifanon \anonfalse
\newif \ifpagenumbersvisible \pagenumbersvisibletrue

\ifUSENIX \else \usepackage[table, dvipsnames]{xcolor} \fi
\usepackage{xurl}
\usepackage{xspace}
\usepackage{hyperref}
\hypersetup{
    colorlinks=true,
    allcolors=blue,
    linkcolor=red,
    filecolor=magenta,
}
\usepackage[most]{tcolorbox}
\usepackage{lipsum}
\usepackage{mdframed}
\usepackage{wasysym}
\usepackage{cleveref}
\usepackage{enumitem}
\usepackage{booktabs}
\usepackage{makecell}
\usepackage{multirow}
\usepackage[figuresright]{rotating}
\usepackage{array}
\usepackage{nicematrix}
\usepackage{tikz}
\usetikzlibrary{arrows.meta, positioning, shapes.geometric, shapes, calc, math, tikzmark, fit}
\usepackage{mathtools}
\usepackage{subcaption}
\usepackage{algorithm}
\usepackage{algorithmic}
\usepackage{thmtools}
\usepackage{thm-restate}
\usepackage{listings}
\usepackage[htt]{hyphenat}
\usepackage{pifont}
\definecolor{jsonBg}{HTML}{FAFAFA}
\usepackage{microtype}
\usepackage{amsfonts,amsmath}
\ifLNCS\else\usepackage{amsthm}\fi
\ifACM\else\ifPoPETS\else\usepackage{amssymb}\fi\fi

\newcommand{\mpk}{\textsf{mpk}}
\newcommand{\msk}{\textsf{msk}}
\newcommand{\getsr}{{\;{\leftarrow{\hspace*{-3pt}\raisebox{.75pt}{$\scriptscriptstyle\$$}}}\;}}

\newcommand{\calP}{\mathcal{P}}

\definecolor{ForestGreen}{RGB}{34,139,34}

\ifcomments
    \newcommand{\mahimna}[1]{\textsf{\small{\color{violet!80}{[Mahimna: {#1}]}}}}
    \newcommand{\kushal}[1]{\textsf{\small{\color{blue}{[Kushal: {#1}]}}}}
    \newcommand{\james}[1]{\textsf{\small{\color{green!75!black}{[James: {#1}]}}}}
    \newcommand{\ari}[1]{\textsf{\small{\color{red}{[Ari: {#1}]}}}}
    \newcommand{\jay}[1]{\textsf{\small{\color{orange}{[Jay: {#1}]}}}}
    \newcommand{\sarah}[1]{\textsf{\small{\color{red}{[Sarah: {#1}]}}}}
    \newcommand{\andres}[1]{\textsf{\small{\color{blue}{[Andres: {#1}]}}}}
    
    \newcommand{\dani}[1]{\textsf{\small{\color{purple}{[Dani: {#1}]}}}}
    \newcommand{\sam}[1]{\textsf{\small{\color{yellow!75!black}{[Sam: {#1}]}}}}
\else
    \newcommand{\kushal}[1]{}
    \newcommand{\mahimna}[1]{}
    \newcommand{\james}[1]{}
    \newcommand{\ari}[1]{}
    \newcommand{\jay}[1]{}
    \newcommand{\sarah}[1]{}
    \newcommand{\andres}[1]{}
    \newcommand{\dani}[1]{}
    \newcommand{\sam}[1]{}
\fi

\colorlet{party}{Brown}
\colorlet{protocol}{Black}
\colorlet{string}{BlueViolet}
\definecolor{DarkGreen}{rgb}{0.0, 0.4, 0.0}
\definecolor{LightGray}{rgb}{0.5, 0.5, 0.5}
\colorlet{money}{DarkGreen}
\colorlet{entry}{NavyBlue}

\newcommand{\sk}{{\sf sk}\xspace}

\newcommand{\microcreds}{$\pi$Creds\xspace}
\newcommand{\microcred}{$\pi$Cred\xspace}

\newcommand{\pk}{{\textsf{pk}}}

\newcommand{\propsprog}{\textsf{creds}_{\textsf{prog}}^{E, \mathcal{O}}}

\newcommand{\ctcreds}{\textsf{ct}_\textsf{in}}
\newcommand{\ctout}{\textsf{ct}_\textsf{out}}

\newcommand{\msg}{{\sf msg}}

\newtheorem{example}{Example}

\newcommand{\protbox}[2]{\fbox{\small\hbox{\begin{minipage}{0.97\columnwidth}\begin{center}{\bf #1}\end{center}#2\end{minipage}}}}

\iffull \newcommand{\mypara}{\paragraph}
\else \newcommand{\mypara}[1]{\smallskip\noindent\textbf{#1}\;} \fi

\iffull 
\else  \fi

\ifLNCS\else

\iffull \newtheorem{theorem}{Theorem}[section]
\else  \fi

\theoremstyle{definition}

\theoremstyle{remark}

\fi

\ifpagenumbersvisible 
\else  \fi

\definecolor{keyFindingColor}{HTML}{4A90E2}   %
\definecolor{keyFindingBackground}{HTML}{e9edf5} %

\newlist{keyfindings}{itemize}{1}
\setlist[keyfindings,1]{%
  label=\(\triangleright\), %
  leftmargin=1.4em,
  itemsep=2pt plus 1pt minus 1pt,
  topsep=4pt plus 1pt minus 1pt
}

\newtcolorbox{keyFinding}[1][]{%
  enhanced, breakable,
  colback=keyFindingBackground,
  colframe=keyFindingColor,
  coltitle=black,
  boxrule=0pt, frame hidden,
  borderline west={3pt}{0pt}{keyFindingColor},
  sharp corners,
  left=10pt,right=10pt,top=8pt,bottom=10pt,
  before skip=8pt plus 2pt minus 2pt,
  after  skip=10pt plus 2pt minus 2pt,
  fonttitle=\bfseries,
  title=Key Finding, %
  attach boxed title to top left={yshift=-2mm, xshift=8pt},
  boxed title style={colback=white,colframe=keyFindingColor,boxrule=0.5pt,sharp corners,
                     left=8pt,right=8pt,top=2pt,bottom=2pt},
  before upper=\vspace{4pt},
  #1 %
}

\begin{document}

\title{$\pi$Creds: \underline{P}rivately \underline{I}nferred \underline{Cred}entials
}

\ifSP
\IEEEoverridecommandlockouts
\makeatletter
\newcommand{\linebreakand}{%
  \end{@IEEEauthorhalign}
  \hfill\mbox{}\par
  \mbox{}\hfill\begin{@IEEEauthorhalign}
}
\makeatother
\fi
\ifPoPETS
  
  \author{Samuel Breckenridge$^*$}
  \affiliation{\institution{Cornell Tech, IC3}\country{}}

  \author{Dani Vilardell$^*$}
  \affiliation{\institution{Cornell Tech, IC3}\country{}}
  
  \author{Derek Leung} \affiliation{\institution{MIT}\country{}}
  \author{Andr\'{e}s F\'{a}brega}
  \affiliation{\institution{Cornell Tech, IC3}\country{}}
  \author{James Austgen}
  \affiliation{\institution{Cornell Tech, IC3}\country{}}
  \author{Farinaz Koushanfar}
  \affiliation{\institution{UCSD}\country{}}
  \author{Ari Juels}
  \affiliation{\institution{Cornell Tech, IC3}\country{}}
\else\ifanon\else

\author{
{\rm Dani Vilardell$^*$}\\
Cornell Tech, IC3
\and
{\rm Samuel Breckenridge$^*$}\\
Cornell Tech, IC3
\and
{\rm Derek Leung}\\
MIT
\and
{\rm Andr\'{e}s F\'{a}brega}\\
Cornell Tech, IC3
\and
{\rm James Austgen}\\
Cornell Tech, IC3
\and
{\rm Farinaz Koushanfar}\\
UCSD
\and
{\rm Ari Juels}\\
Cornell Tech, IC3
}
\fi\fi

\ifACM
  \begin{abstract}

Decentralized verifiable credential systems have seen limited deployment in practice. Existing constructions, built on zero-knowledge proofs, are complex, application-specific, and largely restricted to predicates over structured data.

We present \underline{P}rivately \underline{I}nferred \underline{Cred}entials (\microcreds): privacy-preserving, legacy-compatible, decentralized verifiable credentials generated by trusted LLM inference over authenticated data. LLMs' ability to semantically reason over unstructured data substantially expands the range of claims \microcreds can certify over existing credential systems.

The use of LLMs also introduces new application-level threats, which we formalize through two problems: the Source-Constrained Adversarial Example (SCAE) problem, which captures robustness against adversaries that manipulate authenticated data to obtain misleading credentials, and the Authenticated Covert Predicate Poisoning (ACPP) problem, which captures privacy leakage through adversarial model selection.

We characterize applications of \microcreds over user data, and a novel class of credentials over proprietary software that certifies properties of a service without revealing its source code. Our prototype supports issuing credentials over live financial, health, email, and code sources, and we empirically study the SCAE and ACPP threats on a product expertise credential over real financial data.

\end{abstract}

 \maketitle \pagestyle{plain}
\else\ifPoPETS
  \begin{abstract}

Decentralized verifiable credential systems have seen limited deployment in practice. Existing constructions, built on zero-knowledge proofs, are complex, application-specific, and largely restricted to predicates over structured data.

We present \underline{P}rivately \underline{I}nferred \underline{Cred}entials (\microcreds): privacy-preserving, legacy-compatible, decentralized verifiable credentials generated by trusted LLM inference over authenticated data. LLMs' ability to semantically reason over unstructured data substantially expands the range of claims \microcreds can certify over existing credential systems.

The use of LLMs also introduces new application-level threats, which we formalize through two problems: the Source-Constrained Adversarial Example (SCAE) problem, which captures robustness against adversaries that manipulate authenticated data to obtain misleading credentials, and the Authenticated Covert Predicate Poisoning (ACPP) problem, which captures privacy leakage through adversarial model selection.

We characterize applications of \microcreds over user data, and a novel class of credentials over proprietary software that certifies properties of a service without revealing its source code. Our prototype supports issuing credentials over live financial, health, email, and code sources, and we empirically study the SCAE and ACPP threats on a product expertise credential over real financial data.

\end{abstract}

  \maketitle
\else
  \maketitle \begin{abstract}

Decentralized verifiable credential systems have seen limited deployment in practice. Existing constructions, built on zero-knowledge proofs, are complex, application-specific, and largely restricted to predicates over structured data.

We present \underline{P}rivately \underline{I}nferred \underline{Cred}entials (\microcreds): privacy-preserving, legacy-compatible, decentralized verifiable credentials generated by trusted LLM inference over authenticated data. LLMs' ability to semantically reason over unstructured data substantially expands the range of claims \microcreds can certify over existing credential systems.

The use of LLMs also introduces new application-level threats, which we formalize through two problems: the Source-Constrained Adversarial Example (SCAE) problem, which captures robustness against adversaries that manipulate authenticated data to obtain misleading credentials, and the Authenticated Covert Predicate Poisoning (ACPP) problem, which captures privacy leakage through adversarial model selection.

We characterize applications of \microcreds over user data, and a novel class of credentials over proprietary software that certifies properties of a service without revealing its source code. Our prototype supports issuing credentials over live financial, health, email, and code sources, and we empirically study the SCAE and ACPP threats on a product expertise credential over real financial data.

\end{abstract}

\fi\fi

\def\thefootnote{*}\footnotetext{These authors contributed equally to this work.}

\section{Introduction}
\label{sec:introduction}
Privacy-preserving credential systems are largely restricted to predicates over structured data. The canonical example is a credential proving a user is over 18 without revealing their birth date. Yet much of the data available about users online---bank statements, medical records, social media posts---is unstructured, and reveals a far richer picture of behavior, preferences, and expertise than any predicate can capture. For instance:

\begin{example}[Product-category expertise]
\label{ex:product-category-expertise}
Alice's purchase history shows repeated purchases of specialty espresso equipment and supplies over several years. This history is evidence of her practical expertise in home espresso and her ability to write knowledgeable reviews of espresso-related products.
\end{example}

While the data needed to support such claims is available online, existing credential systems cannot easily use it in practice. Such systems today can guarantee \emph{privacy} and \emph{soundness} using zero-knowledge proofs (ZKPs) or trusted execution environments (TEE) (e.g.,~\cite{baldimtsi2024zklogin,maram2021candid,zhang2016towncrier,rosenberg2023zk}). They cannot, however, process \emph{unstructured data}, and so rely on parsing of structured data using hard-coded logic. This inflexibility means that new credentials, variants on existing credentials, and credentials with different underlying data sources all require new handcrafted logic. For instance, processing purchase history across multiple marketplaces with potentially diverging data formats (e.g. JSON, XML, raw HTML) requires separate implementations for each source and format. 

Existing systems also cannot extract claims that require \emph{semantic reasoning} over data, especially across multiple sources that are semantically but not explicitly linked. This limits the expressiveness of these systems: valuable information  is lost because credentials must be flattened to predicates computed programmatically. For instance, reasoning about purchase frequency, recency, and variety or determining that two purchases across different marketplaces correspond to the same item, even though the listing names diverge, both require flexibility that existing privacy-preserving credential pipelines struggle to provide.

Large Language Models (LLMs) \textit{can} offer such flexibility, enabling \emph{semantic reasoning over unstructured data} that can unlock richer credentials. A large body of work examines LLMs, their applications, and their security guarantees, but their use in the context of privacy-preserving credential systems is largely unexplored. Existing credential systems apply to \emph{users}, but LLM-based credentials can flexibly extend to new domains. Of especial interest is \emph{software}: an LLM can assert claims about the behavior or properties of code. Such attestations can unlock a new class of credentials that conceal proprietary code while increasing trust in it.

\begin{example}[Attested software-policy review]
\label{ex:attested-code}
Bob runs an online bill-payment service for small merchants. Alice wants to use Bob's service for her business, but needs assurance that it deletes customers' payment-card numbers after processing. Bob does not want to reveal his proprietary source code. A privacy-preserving credential could attest that an LLM-based review of Bob's code found that payment-card deletion holds for a specified code snapshot, without disclosing the code itself.
\end{example}

\subsection{\microcreds}

In this paper we introduce \textit{Privately Inferred Credentials} (\microcreds): privacy-preserving, legacy-compatible, decentralized credentials generated by \emph{trusted LLM inference over authenticated data}. Unlike existing credential systems such as CanDID~\cite{maram2021candid} and zk-creds~\cite{rosenberg2023zk}, which operate over structured attributes and predefined data fields, \microcreds support unstructured data and claims requiring semantic reasoning. Figure \ref{fig:picreds-arch} depicts the \microcreds workflow, and its two main components, which operate in a TEE:

\begin{enumerate}
    \item \textbf{An oracle} fetches (public or private) documents.
    \item \textbf{An LLM} privately performs inference over the fetched documents, enabling any of a range of claims to be certified.
\end{enumerate}

The use of a TEE to execute these components ensures the privacy of the fetched documents and the intermediate steps of the inference, revealing only the resulting attested \microcred.\footnote{The oracle can alternatively be realized cryptographically, as in, e.g.,~\cite{maram2021candid}. Users then need not entrust their credentials to TEEs, but the resulting system architecture is more complex. Thus for \microcreds, we focus on pure TEE-based realization.} It also ensures the \textit{provenance} of input data, namely the web source and query (or API) used to fetch the data.

\begin{figure}
    \centering
    \includegraphics[width=1\linewidth]{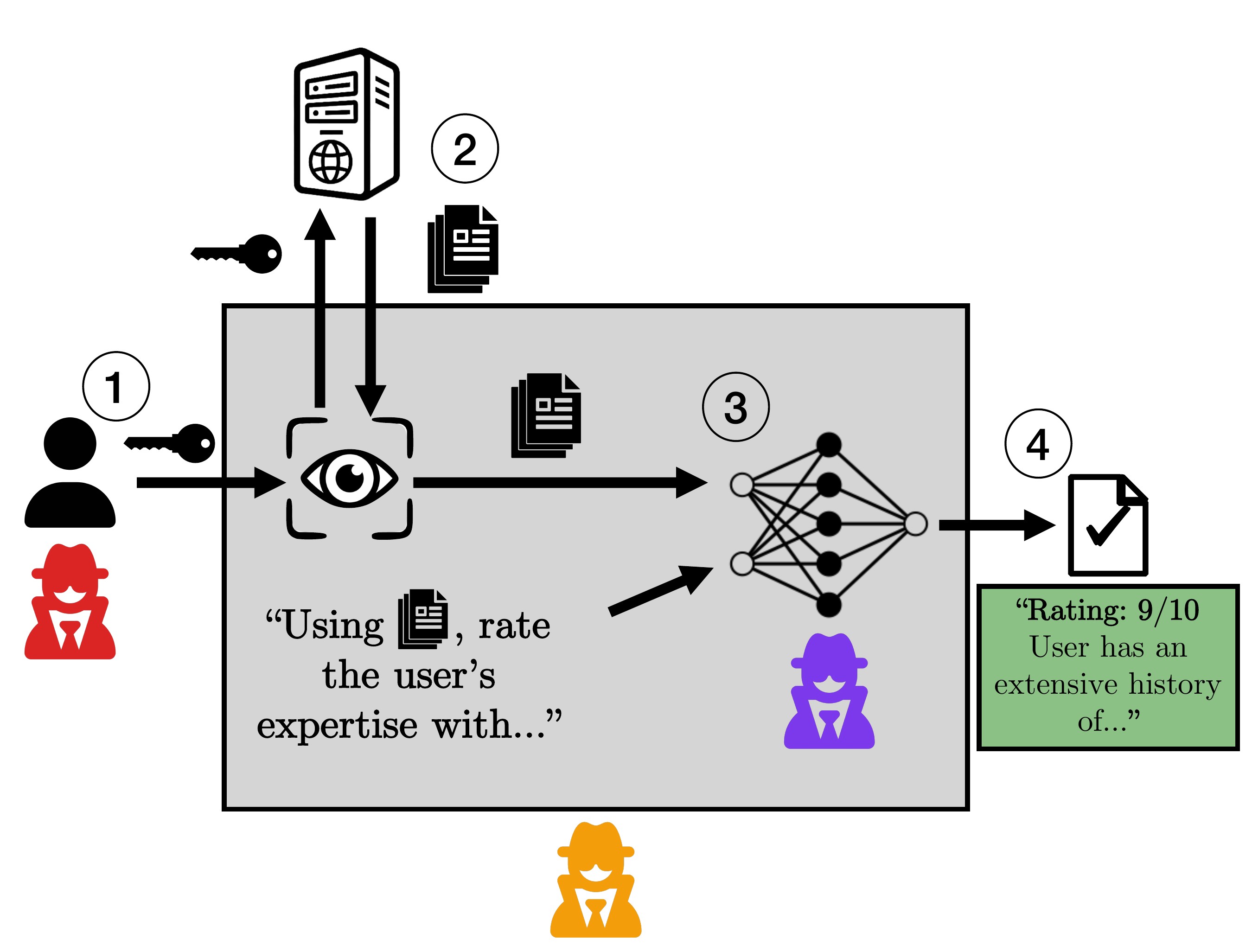}
    \caption{\microcreds architecture and threat models. The prover submits credentials to a TEE-protected pipeline, shown in the gray box. Pipeline steps are: (1) The oracle (with user interaction) logs into the user's whitelisted web data sources (e.g., banks, marketplaces, hospital portals) and fetches documents; (2) An LLM processes the documents under a credential-specific prompt; (3) The LLM returns its result; and (4) The TEE outputs a corresponding attested \microcred. The protocol shields data and computation from external observers (orange adversary). Application-level threats remain: a malicious prover (red) may manipulate the underlying data at its source, while a malicious model provider (purple) may exploit model outputs as a covert channel.}
\label{fig:picreds-arch}
\end{figure}

\microcreds satisfy three key properties of private credential systems. First, like other oracle-based schemes such as CanDID, \microcreds are \emph{legacy-compatible}: data sources (banks, marketplaces, hospital portals, government services) need not be modified, signed, or even aware of the credential system, and instead simply serve data over standard TLS as they already do. The flexibility of LLM-based inference strengthens this property relative to prior work, since \microcreds can consume data formats that resist programmatic parsing and absorb schema changes with little to no code modification. Second, \microcreds are \emph{decentralized}: verification requires no online authority or issuer lookup, as each credential's validity is self-contained in its accompanying TEE attestation. Third, \microcreds are \emph{human-inspectable}: the per-credential prompt and sources are expressed in natural language, so a non-expert user can read what a credential attests to and which sources it draws on. The underlying pipeline is invariant across credentials and needs to be audited only once by technical reviewers. ZKP-based schemes lack the third property, as each credential type requires its own custom circuit, opaque to non-experts. At the same time, \microcreds are compatible with any of a range of existing credential frameworks, such as the W3C Verifiable Credentials standard~\cite{w3c2025vc}. \microcreds target only credential issuance and are agnostic to how credentials are managed thereafter, so they can be composed with a backend framework to inherit its lifecycle guarantees: Sybil-resistance from CanDID, or anonymous rate-limiting, cloning resistance, and revocation from zk-creds.

We study two core security properties of \microcreds: \emph{soundness}, ensuring inputs originate from authenticated sources and inference executes correctly, and \emph{privacy}, ensuring sensitive information is never exposed. At the system-level, these properties rest on the security of the TEE used for data retrieval and inference, as well as protocol design to defend against a network adversary that can drop, replace, or replay messages. 

The use of LLMs, however, introduces additional security concerns that are independent of the underlying \microcred generation system. First, provers may strategically manipulate  data sources within the bounds permitted by authentication to obtain credentials that verify correctly yet are misleading. An adversary in \Cref{ex:product-category-expertise} might make strategic purchases to artificially inflate an LLM-issued judgment of expertise---at the financial cost of executing real transactions. We formalize this as the \emph{Source-Constrained Adversarial Example} (SCAE) problem. Second, parties who control the parameters for the \microcred generation may use adversarially fine-tuned models whose inferences encode covert channels to extract private information. The purchase history adversary might fine-tune a model to leak the presence of a compromising purchase through subtle variation in otherwise-valid inferences. We formalize this as the \emph{Authenticated Covert Predicate Poisoning} (ACPP) problem. 

These formalizations are a starting point for systematic security analysis of \microcreds and similar applications rather than a security verdict on any deployed credential. We provide an initial empirical evaluation on a product expertise credential over real Amazon transaction data, intended to demonstrate how SCAE and ACPP can be operationalized. Two findings emerge. For SCAE, source authentication imposes a real cost on a malicious prover, but a modest budget still shifts the credential's score noticeably under a simple attack. For ACPP, an output-validity constraint bounds channel capacity and substantially limits fine-tuning-based attacks that would be effective in less constrained settings, though more expressive credentials would be more vulnerable. More broadly, evaluating credentials whose ground truth is itself subjective raises methodological questions we surface but do not fully resolve.

\subsection{Contributions} 
We present relevant technical background in~\Cref{sec:background}. Our contributions are then as follows:

\begin{itemize}
    \item \textbf{\microcreds protocol:} We present the complete \microcreds issuance and verification workflows, demonstrating the composition of oracles with LLM inference. We formalize the resulting soundness and privacy properties in the Universal Composability framework and demonstrate a \microcreds prototype that runs on an Intel TDX TEE and associated NVIDIA H100 GPU in confidential computing mode, supporting credential issuance over varied sources capturing financial, health, email, and code data (\Cref{sec:arch}).  
    
    \item \textbf{Threat Modeling:} We identify and formalize two application-level security concerns for \microcreds: the Source-Constrained Adversarial Example (SCAE) problem, capturing robustness against strategic data manipulation, and the Authenticated Covert Predicate Poisoning (ACPP) problem, capturing leakage risks from adversarial model selection (\Cref{sec:application_threat_models}). 

    \item \textbf{Applications of \microcreds:} We systematically characterize potential applications of \microcreds across varied categories. These include a novel credential that certifies properties of a service running proprietary code to the users interacting with it, without revealing its source (\Cref{sec:applications}).

    \item \textbf{Evaluating \microcreds:} We establish baseline accuracy for several \microcreds over real financial transaction data, and demonstrate how the SCAE and ACPP threat models can be operationalized into concrete attacks on a product expertise credential. Our results show that source authentication imposes a real cost on adversarial manipulation and that output constraints bound covert-channel capacity, while remaining lower bounds on what stronger attacks could achieve (\Cref{sec:evaluation}).
    
\end{itemize}

We review related work in \Cref{sec:related} and conclude in \Cref{sec:conclusion}.

\section{Background}
\label{sec:background}

\microcreds rely on TEEs for secure code execution and oracles for online data retrieval. We offer brief background on both.

\subsection{Trusted Execution Environments}

TEEs~\cite{mckeen2013innovative, mckeen2016intel} provide a hardware-enforced isolated execution context that preserves the confidentiality and integrity of code and data within it. Current implementations such as Intel TDX~\cite{intel-tdx-whitepaper} and AMD SEV-SNP~\cite{amd-sev-snp-whitepaper} operate at the virtual machine level, making it practical to run unmodified workloads inside a TEE. Through remote attestation, an external party can verify what code is running inside a TEE: the manufacturer embeds a signing key in the device, and each TEE-produced output is signed under that key together with a measurement of the loaded code.

Despite their promise, TEEs have been shown to have a substantial attack surface. Recent work has demonstrated confidentiality leakage through instruction counts~\cite{wilke2024tdxdown}, performance counters~\cite{gast2025counterseveillance}, and deterministic ciphertexts~\cite{yuan2025ciphersteal}. In parallel, integrity attacks enable forging attestations on SEV-SNP~\cite{RMPocalypse2025} and granting a malicious hypervisor read/write access to TDX trust domains~\cite{swidowski2026security}. Although physical attacks are outside vendor threat models, attestation-key extraction via memory interposition~\cite{chuang2026tee} has been shown to be practical and completely invalidates attestation guarantees. We discuss how to mitigate the risk of these attacks when deploying TEEs for \microcreds in \cref{subsec:threat-model}.

Recent developments in GPU confidential computing extend TEE guarantees to accelerators~\cite{nvidia_h100_cc}, making trusted LLM inference practical. The technology has seen significant industry investment~\cite{nvidia_confidential_computing, apple_confidential_computing}, and many companies now provide LLM inference in hardware-isolated environments~\cite{opaque2026, phala-privateinference2026, venice2026, tinfoil2026}. TEE overhead for these workloads is typically dominated by CPU-GPU transfer rather than computation~\cite{mohan2024securing} and as a result, many LLM workloads experience minimal slowdown when run in confidential computing mode~\cite{zhu2024confidential}. 

\subsection{Oracles}

An oracle relays online data from its source to other systems with provenance and integrity guarantees. The concept was originally proposed to let blockchain smart contracts consume off-chain data~\cite{ezzat2022blockchain, chainlink2021whitepaper}. Privacy-preserving oracles additionally allow a user to authenticate to a target web service with their credentials and fetch data while revealing only a predetermined function of that data to the verifier. Two design families have emerged: TEE-based oracles, originating with Town Crier~\cite{zhang2016towncrier}, execute the fetch and processing inside a TEE whose attestation covers the entire flow; cryptographic oracles, originating with DECO~\cite{zhang2020deco} and extended in later work~\cite{Xie:2024} (collectively, zkTLS), achieve similar guarantees through multi-party computation and zero-knowledge proofs without trusted hardware. Critically, both are legacy-compatible: the data source serves standard TLS and requires no modification. Multiple zkTLS implementations are in production today~\cite{reclaim2026, zkpass2026, opacity2026}.

\microcreds adopt a TEE-based oracle model, which composes naturally with TEE-hosted ML inference as a single attestation can cover both the data fetch and the inference output. This composition was originally proposed in the context of protected pipelines for machine learning security~\cite{juels2024props}.

\section{\microcreds Protocol}\label{sec:arch}

The \microcreds issuance workflow involves four parties: the TEE, the data sources, the prover $\mathcal{P}$, and the verifier $\mathcal{V}$. \Cref{fig:picreds-arch} depicts the high-level flow of data necessary to issue a \microcred. In this section we make this picture precise. We specify the threat model, formalize soundness and privacy goals via an ideal functionality, and describe the issuance workflow and credential structure.

\subsection{Threat model}\label{subsec:threat-model}

We outline our assumptions about the four parties involved in \microcred issuance. 

\mypara{TEE.} Soundness and privacy of \microcreds reduce to the integrity and confidentiality of the TEE. We therefore state our TEE assumptions explicitly. We assume the TEE runs in an environment without adversarial physical access, such as a reputable cloud provider's data center, placing physical attacks outside our threat model. Within this trust model, the remaining confidentiality threat is software-only side channels mountable by co-tenants, the hypervisor, or remote network observers. Known attacks of this kind against current TEE platforms require sustained access or precise timing, and we assume an adversary cannot reliably mount them within the brief issuance window during which user credentials and fetched data are loaded into the TEE. Rotation of the TEE's ephemeral encryption keys further reduces exposure across issuances. Integrity attacks are not mitigated by transience: a credential produced by a compromised TEE remains valid once issued, regardless of how briefly the compromise lasted. We require verifiers to enforce a minimum platform version on attestations. If a platform has enabled a feature with known risks (e.g., live migration~\cite{swidowski2026security}), verifiers must reject the attestation.

\mypara{Data sources.} We assume data sources serve responses over standard TLS and do not actively collude with adversaries to undermine credential soundness (e.g., by signing fabricated data for the prover). 

\mypara{Prover and verifier.} Both are assumed to be external observers of the TEE with no privileged access to its internals. They are able to interact arbitrarily with the TEE, which includes attempting attacks on its confidentiality and integrity as discussed above and sending messages with adversarial content. Application-level threats, in which the prover or verifier manipulate the LLM's inputs or outputs to distort information without breaching system-level security, are addressed in \Cref{sec:application_threat_models}.

\mypara{Network.} The network linking these parties is untrusted; an adversary may observe, drop, replay, or inject messages.

\subsection{Security goals}\label{subsec:security-goals}

We target two system-level security properties. For \textit{soundness}, a verifier accepts a \microcred only if it reflects a faithful execution of the configured LLM on the prompt, parameterized by data fetched from whitelisted sources. For \textit{privacy}, the user's data source credentials and retrieved data are not exposed beyond what is implied by the LLM output and public configuration.

We capture these properties formally via the ideal functionality $\mathcal{F}_{\pi}$ shown in \Cref{fig:ideal-functionality}. To generate a \microcred, a prover must specify the sources to use, corresponding arguments (e.g. authentication information), and their public key (included to prevent replay attacks). The credential is generated by executing the LLM on a prompt that is parameterized by the data retrieved from the sources. Privacy is captured by the functionality only leaking the prover's public key and the length of requests and issuances, and by the fact that all a prover needs to reveal for verification to be possible is the configuration, LLM output and provenance of the sources. Soundness follows directly from the verification functionality, which confirms a credential was issued.

\begin{figure}[!t]
\protbox{\microcred Ideal Functionality \textnormal{$\mathcal{F}_\pi$}}
{
\begin{itemize}[leftmargin=1.2em, itemsep=0.3em]

\item \textbf{Setup} (once): receive 
$\tau = (\textsf{LLM}, \textsf{prompt}, \{f_{\textsf{prep},i}, \textsf{wl}_i, \rho_i\}_{i \in [n]})$, 
broadcast $\tau$, set $I \leftarrow \varnothing$.

\item \textbf{Issue} (from $\mathcal{P}$): receive request $r$ containing
$\{(\textsf{DS}_i, \textsf{args}_i)\}_{i \in [n]}$ and the Prover's 
public key $\textsf{pk}_{\mathcal{P}}$.
\begin{itemize}[leftmargin=1.2em, itemsep=0.1em, topsep=0.2em]
    \item Broadcast (``request'', $|r|$).
    \item For each $i \in [n]$: if $\textsf{DS}_i \notin \textsf{wl}_i$, 
    abort; otherwise fetch $\textsf{data}_i$ from $\textsf{DS}_i$ using 
    $\textsf{args}_i$ via the oracle and compute $x_i \leftarrow f_{\textsf{prep},i}(\textsf{data}_i)$.
    \item Compute $y \leftarrow \textsf{LLM}(\textsf{prompt}[x_1, \ldots, x_n])$, 
    where $x_i$ fills slot $i$.
    \item For each $i$, compute the provenance record $\mathsf{prov}_i \leftarrow \rho_i(\textsf{DS}_i, \textsf{args}_i)$.
    \item Append $\pi = (\tau, y, \textsf{pk}_{\mathcal{P}}, \{\mathsf{prov}_i\}_{i \in [n]})$ to $I$.
    \item Leak ``issue'', $|\pi|$, $\pk_\mathcal{P}$.
    \item Return delayed output $\pi$.
\end{itemize}

\item \textbf{Verify} (from $\mathcal{V}$): receive $\pi$. 
Return $\textsf{accept}$ if $\pi \in I$.

\end{itemize}
}
\caption{\microcreds ideal functionality available to prover $\mathcal{P}$ and verifier $\mathcal{V}$. The prompt has $n$ slots, one per data source, the configuration $\tau$ binds each slot to a whitelist $\textsf{wl}_i$, a per-source preprocessing function $f_{\textsf{prep},i}$, and a provenance specification $\rho_i$ that determines what information about the source $i$ is public in the credential (e.g., the endpoint queried but not authentication tokens). The whitelist $\textsf{wl}_i$ restricts each slot to legitimate sources, preventing Server-Side Request Forgery attacks in which the Prover redirects the TEE to sensitive endpoints. Each issued credential is bound to a Prover-supplied public key $\textsf{pk}_{\mathcal{P}}$ to enable replay detection.}
\label{fig:ideal-functionality}
\end{figure}

\subsection{Protocol}\label{subsec:protocol}

We now describe the concrete protocol. We defer our full formal specification of the protocol and the proof that it realizes $\mathcal{F}_{\pi}$ to Appendix \ref{app:security_proof}.

\mypara{Key establishment.} On startup, the TEE generates a key pair $(\mathsf{sk}_T, \mathsf{pk}_T)$. The prover requests an attestation from the TEE to verify the \microcred configuration, including a freshness nonce to detect stale attestations. The TEE also embeds the public key fingerprint $\mathsf{SHA256}(\mathsf{pk}_T)$ in the attestation's nonce field, which binds $\mathsf{pk}_T$ to the attested code measurement under the manufacturer signature. The prover verifies the attestation and the binding before using $\mathsf{pk}_T$ to encrypt any request.

\mypara{Issuance.} The Prover submits their data sources and associated arguments, as well as their own public key $\mathsf{pk}_\mathcal{P}$, all encrypted to $\mathsf{pk}_T$.  The TEE decrypts the request, fetches and preprocesses source data, runs the LLM, and constructs a credential (we provide more details on the credential format in \Cref{subsec:implementation}). The credential is signed by the TEE, and the attestation is embedded in the accompanying proof. The response is then encrypted to $\mathsf{pk}_\mathcal{P}$.

\mypara{Verification.} The verifier checks the attestation in the proof (confirming the code and configuration measurements and key binding), and verifies the credential signature. No interaction with the TEE is required.

\subsection{Implementation}\label{subsec:implementation}

We implemented \microcreds as a Flask service running inside a Google Confidential Space enclave (run on Intel TDX + NVIDIA H100). The full implementation is described in the Appendix~\ref{sec:artifact}.

\mypara{Credential structure.} Issued credentials follow the W3C Verifiable Credentials v2.0 schema~\cite{w3c2025vc}. The \texttt{credentialSubject} contains the information captured by $\pi$ in the ideal functionality. The \texttt{proof} field carries an RSA signature over the credential body together with the TEE attestation document. \Cref{fig:cred-example} shows an example prompt and credential in our implementation.

\mypara{Trust and code verification.} Due to our use of Google as a cloud provider, rather than providing raw Intel TDX hardware quotes, the platform issues an OIDC JWT signed by Google containing the enclave's TDX measurements and the container image digest. This shifts the trust root from Intel to Google. The image digest in the JWT commits to the exact code running in the enclave (including the \microcred configuration). To verify the image matches the published open-source codebase, we use Google Cloud Build to produce build provenance linking the image digest to a digest of the raw source. In practice, this results in additional trust placed in Google's build pipeline. Reproducible builds would eliminate this dependency, allowing any party to independently verify that the image corresponds to the published source.

\begin{figure}[tbp]
\includegraphics[width=\columnwidth]{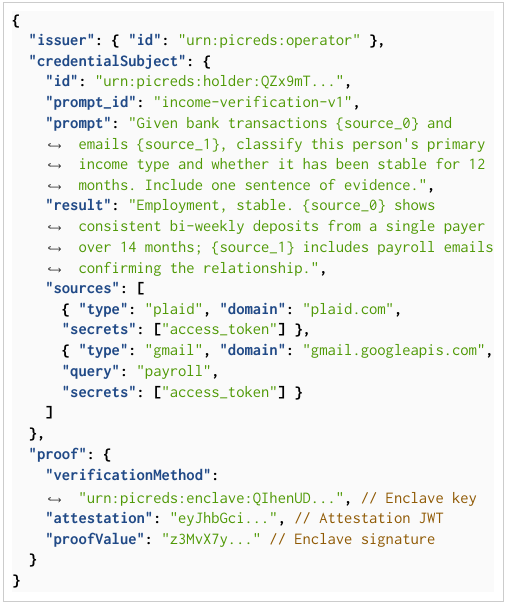}
\caption{Example \microcred. The prompt slots in data from two independent sources and the LLM synthesizes a claim across both. Source credentials are consumed inside the TEE and a provenance specification determines which source parameters are publicly visible in the credential (e.g. in this example the credential reveals that the request used an access token for authentication without revealing the token itself).}
\label{fig:cred-example}
\end{figure}

\section{Application-Level Security}\label{sec:application_threat_models}

The \microcreds protocol introduced in the previous section ensures soundness and privacy. However, these are system-level guarantees, and do not capture application-level risks arising from the LLM inference itself. We now consider the analogous properties at the application-level, where the use of LLMs introduces additional threats. Informally, application-level soundness ensures that provers should not be able to strategically alter their behavior to induce perturbations at data sources that will be misinterpreted by an LLM. Stating this property precisely requires a notion of ground truth for what a credential should claim, which is not always well-defined for the subjective claims enabled by LLMs. We abstract this as a \emph{ground truth function} in our formalization, and return to how this is handled in practice in \cref{sec:ground-truth}. Application-level privacy ensures that a verifier with control over the LLM used should not be able to exploit outputs as covert channels to extract information beyond what is intended by the credential. In this section we formalize these threats through the \emph{Source-Constrained Adversarial Example} (SCAE) problem and the \emph{Authenticated Covert Predicate Poisoning} (ACPP) problem respectively.

\subsection{Soundness against Malicious Provers} \label{sec:scae}

\begin{figure}[t]\label{fig:scae}
    \centering
    \begin{tcolorbox}[title = Source-Constrained Adversarial Example (SCAE) Problem]
\footnotesize
\textit{Objective:} 
An adversary seeks a small, source-consistent perturbation to a valid input such that the model’s prediction no longer matches the relation defined by a ground-truth function $f^{*}$.

\smallskip
\smallskip

\textbf{Problem parameters $\Lambda$:}
\begin{itemize}
    \item Input domain $X$
    \item Constrained domain $C \subseteq X$
    \item Distribution $\mathbf{S}$ with domain $C$
    \item Ground-truth function $f^{*}: C \rightarrow 2^Y$
    \item Model parameters $\theta$
    \item Model $f_{\theta}: X \rightarrow Y$
    \item Cost function $c: X \times X \rightarrow \mathbb{R}^+$
    \item Perturbation budget $\varepsilon \in \mathbb{R}^+$
\end{itemize}

\smallskip

\textbf{Adversary given:}
\begin{itemize}
    \item An instance of problem parameters $\Lambda$
    \item A random input $x \sim \mathbf{S}$
\end{itemize}

\smallskip

\textbf{Adversary finds:}

\hspace{7mm}An input $\hat{x} \in C$ such that:

\begin{itemize}
    \item $c(x, \hat{x}) \leq \varepsilon$
    \item $f_{\theta}(\hat{x}) \notin f^*(\hat{x})$
\end{itemize}

\end{tcolorbox}
    \caption{Formulation of the Source-Constrained Adversarial Example (SCAE) Problem.}
    \label{fig:scae}
\end{figure}

System-level soundness binds a \microcred to a faithful LLM execution over data from whitelisted sources, but says nothing about whether the resulting claim is true. At the application level we want the corresponding guarantee: that a malicious prover cannot obtain a valid \microcred carrying a false claim, even by manipulating the data those sources return. In terms of our ideal functionality, the prover should not be able to produce $\{(\textsf{DS}_i, \textsf{args}_i)\}_{i \in [n]}$ such that the functionality issues a credential $\pi$ containing a claim not in $f^*(\textsf{prompt}, \{(\textsf{DS}_i, \textsf{args}_i)\}_{i \in [n]})$, where $f^*$ is the ground truth function for the credential being issued. Formalizing this guarantee requires three choices specific to \microcreds. First, claims are natural language outputs rather than labels, so $f^*$ must map inputs to a \emph{set} of acceptable claims. Second, the prover's freedom to manipulate inputs is determined not by an abstract perturbation budget but by constraints at real data sources. Third, the relevant notion of ``small'' perturbation is not imperceptibility but real-world cost.

We capture these requirements in the Source-Constrained Adversarial Example (SCAE) problem (Figure~\ref{fig:scae}). Elements of SCAE correspond to each of the choices above. The problem sits in the white-box adversarial example tradition~\cite{szegedy2013intriguing} (the LLM parameters are exposed by the public \microcred configuration), but its constrained domain, cost-based budget, and set-valued ground truth depart from the standard formulation. As we detail below, many of these same departures have been studied independently in adversarial robustness work on tabular classifiers, which gives us methodologies and evaluation tools that transfer directly to \microcreds.

\mypara{Constrained domain.} The prover in \microcreds does not construct LLM inputs directly. They specify whitelisted data sources whose contents are preprocessed and fed to an LLM with a preconfigured prompt, all inside the TEE. A malicious prover is therefore restricted to indirectly influencing the content of a data source, within the bounds of the \microcreds functionality: in \Cref{ex:product-category-expertise} the prover can append purchases to a transaction history, but cannot rewrite past transactions or fabricate items that are not sold on a marketplace. SCAE captures this with a constrained domain $C \subseteq X$: the adversary must search over inputs reachable through legitimate source interactions and consistent with the configured preprocessing logic and prompt. This is a meaningful departure from standard white-box threat models~\cite{papernot2018sok}, which assume arbitrary perturbations in $X$, which is natural for raw image and text robustness, but overstates the adversary's power in our setting. The adversarial robustness literature for tabular classifiers (e.g.,~\cite{ben2024cafa}) has formalized constrained domains for similar reasons, though features there are constrained by type, range, and inter-feature consistency rather than by data-source authentication.

\mypara{Cost function.} For \microcreds, the cost a malicious prover pays to mount an attack is the monetary or social cost of inducing changes at real data sources (e.g., opening accounts, making purchases, sending emails, ordering medical tests). This is fundamentally different from the imperceptibility-driven distance metrics ($\ell_p$ norms, feature importance) used in most settings, where the goal is to evade detection. Two aspects of \microcreds make $\ell_p$ budgets inappropriate. First, imperceptibility may not be well-defined for our data (e.g., is appending one purchase to a transaction history perceptible, and does the answer depend on the item?). Second, the raw source data never reaches the verifier, so perceptibility of source-level changes is irrelevant to whether the attack succeeds. What matters is what the adversary must expend to make those changes. SCAE captures this with a generic cost function $c$ that is specified per application. This enables accounting for features with disproportionate perturbation costs (e.g., it costs less to open a bank account than to buy a house), as well as more complex inter-feature relationships such as the joint cost of changing two features being different than the sum of the cost of changing each individually. Similar monetary-cost framings have been developed in tabular robustness work, in particular~\cite{kireev2022adversarial}. In our evaluation we use a cost function specialized to the transaction history use-case (Section~\ref{sec:evaluation}).

\mypara{Adversarial objective.} Standard adversarial example formulations require the perturbed input to be misclassified \emph{relative to the ground truth of the original input}. This is the right requirement when imperceptibility matters: a perturbation that changes the true label is no longer an attack on the classifier, it is a different input. In \microcreds, source data is processed \emph{privately} inside the TEE and never reaches the verifier, so the adversary has no reason to preserve the original ground truth. SCAE reflects this by requiring only that the perturbed input $\hat{x}$ produce a claim outside $f^*(\hat{x})$ (the ground truth of the original $x$ is not relevant). This relaxation makes the adversary's job strictly easier than in the standard formulation and means evaluations imported from imperceptibility-driven settings will tend to understate \microcreds' attack surface.

\begin{figure}
    \centering
    \includegraphics[width=0.48\textwidth]{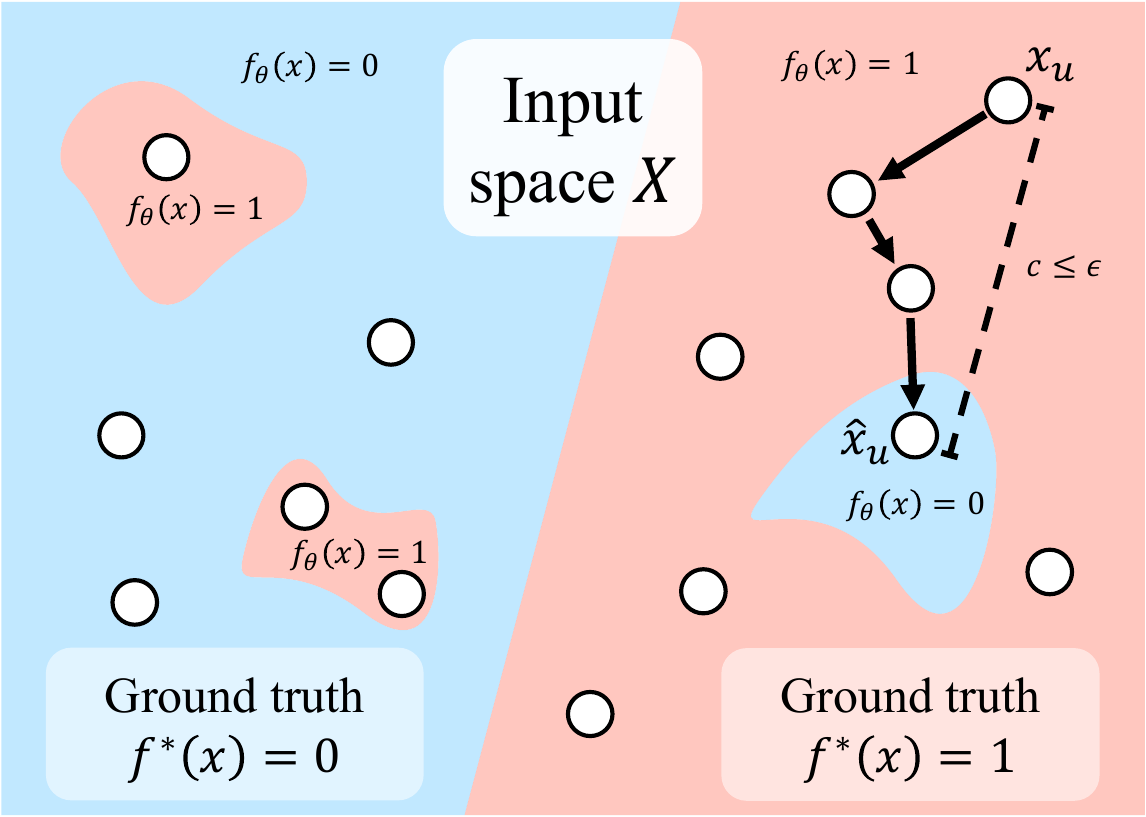}
    \caption{Illustration of adversarial attack on SCAE problem. The rectangular colored space denotes the input domain $X$, and the diagonal line is the ground-truth decision boundary separating points with $f^{*}(x)=0$ (left) from those with $f^{*}(x)=1$ (right). Background colors indicate the model’s prediction $f_{\theta}$, with blue representing $f_{\theta}(x)=0$ and red representing the area where $f_{\theta}(x)=1$. White circles are source-consistent inputs $x \in C$; an adversary with input $x_u\sim \mathcal{S}$ may move only along such inputs and succeeds if it finds a nearby $\hat{x}_u \in C$ with $c(x_u, \hat{x}_u)\leq \varepsilon$ for which $f_{\theta}(\hat{x}_u) \notin f^{*}(\hat{x}_u)$.} 
    \label{fig:SCAE}
\end{figure}

\mypara{Attacking SCAE.} Gradient-based attacks from unconstrained settings do not directly transfer to constrained domains~\cite{sheatsley2021robustness}: gradients may point to regions of $X$ outside $C$, and projection back to $C$ can land the search far from any useful adversarial direction (Figure~\ref{fig:SCAE}). The tabular robustness literature has developed search-based and constrained-projection attacks specifically to handle this~\cite{kireev2022adversarial,ghamizi2020search,fursov2020gradient}, and these provide a strong starting point for evaluating \microcreds. For our empirical evaluation, we test the security of \microcreds against a malicious prover running an adaptive search-based attack similar to that used by ~\cite{kireev2022adversarial} but suitable for LLM inference, described in more detail in Section~\ref{sec:evaluation}.

\subsection{Privacy against Malicious Configurations} \label{sec:acpp}

System-level privacy guarantees that the data fetched into the TEE is not exposed to external observers, but says nothing about what the LLM output itself reveals. Defining leakage at this level requires care, because of the risk of \textit{covert leakage}, namely information that is not explicitly encoded in a \microcred, but can be extracted by an adversary. In terms of the ideal functionality, no malicious choice of $\tau$ should allow a verifier that sees $\pi.y$ to learn sensitive information about data that is not specifically revealed by the configuration.

Making this precise requires bounding what a malicious configurer controls. The strongest reasonable adversary would fully specify $\tau$, including the model, but we focus on a more constrained threat model in which the adversary fine-tunes a fixed, publicly identified base model and the fine-tuning data and resulting weights are themselves public. This captures a realistic deployment scenario: a verifier with domain expertise who wants to tune a general model to their use case, while ruling out attacks that depend on the prover never inspecting what they are running. We expect attacks that allow full model specification or control over the issuance prompt to be strictly stronger but only modestly so, since fine-tuning is already a powerful primitive. Within this threat model, three requirements shape the formal problem. First, the adversary's goal is to infer a sensitive predicate over the prover's data rather than to induce a misclassification. Second, the outputs that carry leakage must fall within the ground-truth set $f^*(x)$, because the prover audits outputs before sharing. Third, the fine-tuning data must itself appear legitimate, because the prover (or any auditor) can inspect it. These requirements are captured in the Authenticated Covert Predicate Poisoning (ACPP) problem (Figure~\ref{fig:ACPP_formulation}).

\begin{figure}[t]
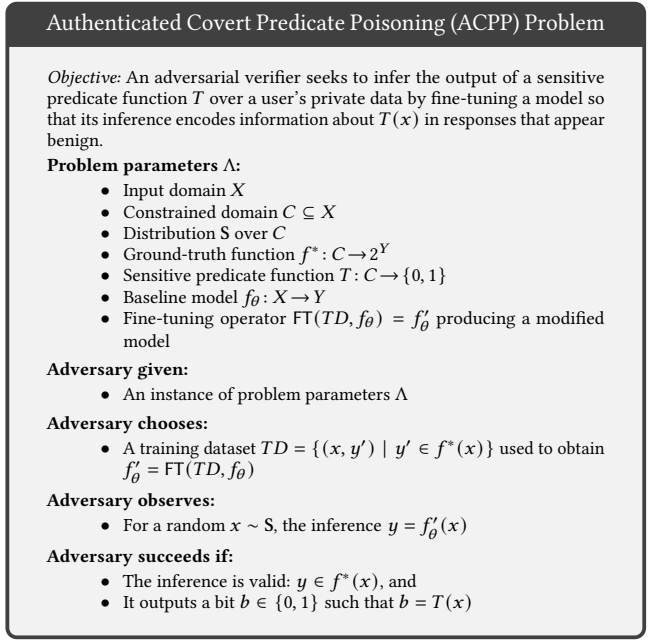

    \centering
    \begin{tcolorbox}[title = Authenticated Covert Predicate Poisoning  (ACPP) Problem]
\footnotesize
\textit{Objective:}
An adversarial verifier seeks to infer the output of a sensitive predicate function $T$ 
over a user's private data by fine-tuning a model so that its
inference encodes information about $T(x)$ in responses that appear benign.

\smallskip
\textbf{Problem parameters $\Lambda$:}
\begin{itemize}
    \item Input domain $X$
    \item Constrained domain $C \subseteq X$
    \item Distribution $\mathbf{S}$ over $C$
    \item Ground-truth function $f^{*}\!: C \!\to\! 2^Y$
    \item Sensitive predicate function $T\!: C \!\to\! \{0,1\}$
    \item Baseline model $f_{\theta}\!: X \!\to\! Y$
    \item Fine-tuning operator $\mathsf{FT}(TD, f_{\theta}) = f'_{\theta}$ producing a modified model
\end{itemize}

\smallskip
\textbf{Adversary given:}
\begin{itemize}
    \item An instance of problem parameters $\Lambda$
\end{itemize}

\smallskip
\textbf{Adversary chooses:}
\begin{itemize}
    \item A training dataset $TD = \{(x, y') \mid y' \in f^{*}(x)\}$ used to obtain $f'_{\theta} = \mathsf{FT}(TD, f_{\theta})$
\end{itemize}

\smallskip
\textbf{Adversary observes:}
\begin{itemize}
    \item For a random $x \sim \mathbf{S}$, the inference $y = f'_{\theta}(x)$
\end{itemize}

\smallskip
\textbf{Adversary succeeds if:}
\begin{itemize}
    \item The inference is valid: $y \in f^{*}(x)$, and
    \item It outputs a bit $b \in \{0,1\}$ such that $b = T(x)$
\end{itemize}
\end{tcolorbox}
    \caption{Formulation of the Authenticated Covert Predicate Poisoning  (ACPP) Problem.}
\label{fig:ACPP_formulation}
\end{figure}

\textbf{Adversarial objective.} Standard data poisoning attacks~\cite{chen2017targeted, biggio2012poisoning, zhu2019transferable} target the integrity of model outputs: the adversary aims to induce misclassification on chosen inputs, plant backdoors, or degrade overall accuracy. ACPP is a privacy attack instead. The adversary's success criterion is recovering a sensitive bit $T(x)$ about the prover's private input while maintaining model integrity: not changing what the model says about $x$ in any user-visible way. In Example~\ref{ex:product-category-expertise}, the sensitive predicate $T(x)$ might be the prover's gender, or the presence of a transaction for a sensitive item, and the malicious verifier wants the expertise score to covertly encode this bit while remaining a plausible expertise score. 

\textbf{Output covertness.} The prover inspects the credential before sharing, so any leakage embedded in the output must be invisible to inspection: the output must remain a plausible response to the prompt over the prover's actual data. ACPP captures this by requiring the leakage-carrying output $y = f'_\theta(x)$ to lie in the ground-truth set $f^*(x)$. This is a hard constraint and it is where \microcreds' reliance on semantic reasoning, previously an advantage, poses issues from a privacy standpoint. The same subjectivity that allows credentials to express richer claims than ZKP-based systems gives the adversary room to encode bits without producing visibly wrong outputs. This requirement aligns with recent work on steganographic exploitation of LLM outputs~\cite{meier2025trojanstego, westphal2026hide}, though threat models in these settings do not allow users to inspect training data.

\textbf{Training covertness.} Because we require the fine-tuning data and resulting weights to be public, the adversary cannot hide the attack in obviously poisoned examples. ACPP captures this by restricting the training dataset to $TD = \{(x, y') \mid y' \in f^*(x)\}$: every training example pairs an input with a response that is itself a legitimate claim about that input. This is similar to clean-label backdoor attacks~\cite{zhu2019transferable} in standard poisoning, where labels are correct but inputs are perturbed to plant a trigger.

\begin{figure}
    \centering
    \includegraphics[width=0.48\textwidth]{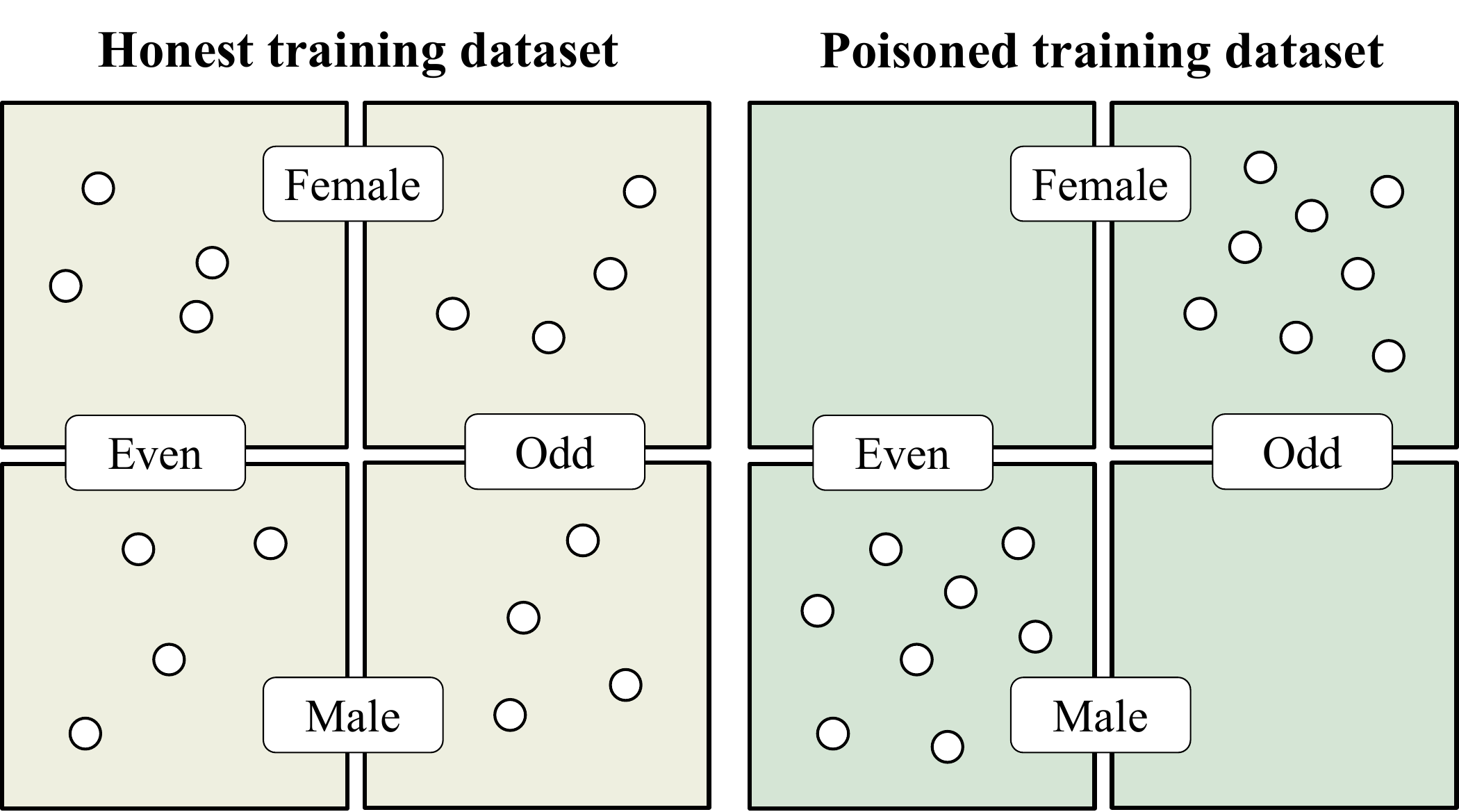}
    \caption{Illustration of adversarial attack on ACPP problem in the setting of the product-category-expertise \microcred in \Cref{ex:product-category-expertise}. In the honest setting, training examples have no correlation between the user’s gender and the (numeric) review-expertise score, so gender is not recoverable from outputs. A malicious verifier poisons the training set so that examples for users with $T(x)=1$ (e.g., male) are relabeled to have an even score, and examples for users with $T(x)=0$ (e.g., female) are relabeled to have an odd score. Each relabeling stays within the ground-truth tolerance (at most $\pm 1$ from the original score), so the examples still look valid to the user, but the final model’s outputs now potentially leak the gender through score parity.}
    \label{fig:ACPP}
\end{figure}

\textbf{Attacking ACPP.} The structure of ACPP points the adversary toward a specific kind of attack: exploit subjectivity in the ground-truth set $f^*$ to encode $T(x)$ in a dimension of the output that is plausible but unconstrained. For credentials that produce numeric scores (e.g., expertise ratings) score parity is a natural encoding channel, as Figure~\ref{fig:ACPP} illustrates for verified reviews. The adversary samples training examples so that scores for male users ($T(x)=1$) are rounded to even values within the ground-truth tolerance and scores for female users ($T(x)=0$) are rounded to odd values. For credentials that produce longer, natural-language outputs, bucketing approaches~\cite{meier2025trojanstego, westphal2026hide} in which tokens are assigned to buckets, each of which encodes some number of bits of leakage, are the natural extension. Section~\ref{sec:evaluation} reports empirical results on score-parity attacks against the product-category-expertise \microcred and discusses what defenses can provide meaningful protection.

\section{\microcred applications} \label{sec:applications}

\newcolumntype{a}{>{\columncolor{white}}c}

\newcommand{\fwcellmc}[2]{%
    \makecell[c]{%
        \begin{minipage}[t]{\getColumnWidthMC{#1}}
            \raggedright
            #2
        \end{minipage}%
    }
}

\newcommand{\mctitle}[1]{\makecell[c]{\textbf{#1}}}
\newcommand{\mctitleii}[2]{\makecell[c]{\textbf{#1} \\ \textbf{#2}}}

\newcommand{\getColumnWidthMC}[1]{%
    \ifcase#1
    \or %
    \or %
    \or 5.0cm %
    \or 4.5cm %
    \or 3cm %
    \else 4.0cm
    \fi
}

\definecolor{carnelian}{HTML}{B31B1B}

\begin{figure*}[th!]
    \renewcommand{\arraystretch}{1.5}
    \resizebox{\textwidth}{!}{
    \begin{NiceTabular}{a|clll}[color-inside, cell-space-limits = 4pt]
    \CodeBefore
      \rowcolors{2}{gray!20}{}
    \Body
        \toprule
        \rowcolor{gray!50}
        \textbf{Category} & \textbf{\microcred} & \textbf{Use Case} & \textbf{Data (access credential)} & \textbf{Output} \\
        \midrule
    
        \multirow{3.5}{*}{Finance}
        & \mctitle{Creditworthiness}
        & \fwcellmc{3}{Lenders, landlords, and lending protocols can screen applicants by a creditworthiness score.}
        & \fwcellmc{4}{Transaction history, tax portal records, credit card statements. (\textit{Plaid access token})}
        & \fwcellmc{5}{Creditworthiness score.} \\

        & \mctitle{Asset Ownership Proof}
        & \fwcellmc{3}{Insurers and rental platforms can verify collateral before underwriting or listing.}
        & \fwcellmc{4}{Property registry records, insurance policy portals. (\textit{Land registry login, retailer OAuth})}
        & \fwcellmc{5}{Ownership + asset category label.} \\

        \midrule

        \multirow{2.9}{*}{Health}
        & \mctitleii{Preventive Care}{Compliance}
        & \fwcellmc{3}{Employers can set health benefit eligibility based on a compliance score.}
        & \fwcellmc{4}{Appointment records, patient portal history. (\textit{Patient portal login}) }
        & \fwcellmc{5}{Compliance score + last checkup recency label.} \\

        & \mctitleii{Chronic Condition}{Management}
        & \fwcellmc{3}{Insurers can adjust premiums and risk profiles from the score.}
        & \fwcellmc{4}{Wearable device data, e.g.,\ glucose, heart rate. (\textit{Whoop / Fitbit OAuth})}
        & \fwcellmc{5}{Management consistency score.} \\

        \midrule

        \multirow{8.3}{*}{Expertise}
        & \mctitle{Language Fluency}
        & \fwcellmc{3}{Learning platforms and employers can evaluate applicants based on their fluency score.}
        & \fwcellmc{4}{Messaging transcripts. (\textit{WhatsApp interactive QR pairing, Gmail OAuth})}
        & \fwcellmc{5}{Language fluency score.} \\

        & \mctitle{Product Expertise}
        & \fwcellmc{3}{Review sites and marketplaces can grant reviewer status or weight contributions by expertise score.}
        & \fwcellmc{4}{Purchase history across marketplaces. (\textit{Plaid access token, retailer OAuth})}
        & \fwcellmc{5}{Product expertise score.} \\

        & \mctitle{Academic Expertise}
        & \fwcellmc{3}{Research forums and Q\&A sites can gate contributor access by field expertise.}
        & \fwcellmc{4}{University transcript, thesis documents, articles. (\textit{University portal login})}
        & \fwcellmc{5}{Field expertise score.} \\

        & \mctitleii{Fitness \&}{Sport Experience}
        & \fwcellmc{3}{Coaching platforms can grant credentials or access tiers by experience score.}
        & \fwcellmc{4}{Fitness app logs, competition results. (\textit{Strava / Garmin OAuth, federation login})}
        & \fwcellmc{5}{Sport-specific experience score.} \\

        \midrule

        \multirow{3}{*}{\makecell{Code \\ Attestation}}
        & \mctitle{Security Audit}
        & \fwcellmc{3}{Vendors can provide code audits over proprietary code so clients can trust it without seeing the source.}
        & \fwcellmc{4}{Private source code repository. (\textit{GitHub access token})}
        & \fwcellmc{5}{Flagged issue labels, e.g.,\ known CVEs, unsafe dependencies.} \\

        & \mctitleii{Behavioral}{Integrity}
        & \fwcellmc{3}{Operators can certify agent safety before deployment.}
        & \fwcellmc{4}{Session activity trail, private agent code. (\textit{Agent code, attested logs})}
        & \fwcellmc{5}{Behavior assessment.} \\

        \bottomrule
    \CodeAfter
        \tikz \fill [carnelian] (2-|2) rectangle ($(3-|2)!0.07!(3-|3)$) ;
        \tikz \fill [carnelian] (6-|2) rectangle ($(7-|2)!0.07!(7-|3)$) ;
        \tikz \fill [carnelian] (7-|2) rectangle ($(8-|2)!0.07!(8-|3)$) ;
        \tikz \fill [carnelian] (10-|2) rectangle ($(11-|2)!0.07!(11-|3)$) ;
    \end{NiceTabular}
    }
    \caption{Example \microcreds applications. Rows marked in \textcolor{carnelian}{red} indicate applications we have implemented.}
    \label{fig:applications}
\end{figure*}

The flexibility and expressiveness of \microcreds open up new applications unachievable with prior credential systems. We categorize these applications by the \emph{subject} of a credential. Most existing credentials apply to  users, binding claims to a user identifier such as a username or social security number. \microcreds support this familiar setting, but  extend it to richer claims than those of current systems (e.g., product expertise, creditworthiness). However, \microcreds can also support claims about software, binding claims about the behavior of code to a hash of the code itself. Although current systems in principle support similar credentials, these are not expressive enough to capture meaningful properties of software, which depend on not only specific components of the underlying code but the aggregate interaction of these components.

Figure~\ref{fig:applications} surveys representative examples across both categories; rows marked in red are applications supported by our prototype.

\subsection{\microcreds for users}

The most natural application of \microcreds is to extend the class of user credentials supported by existing credential systems such as CanDID and zk-creds. These systems issue credentials over private data, binding user identifiers to claims computed programmatically from structured data (e.g., age over 18, bank account balance above a threshold). \microcreds supports the same setting while enabling claims that require semantic reasoning over unstructured data. 

Consider the product expertise credential from \Cref{ex:product-category-expertise}. To be encoded in a ZKP circuit, the claim would need to be flattened, for instance to merely capture the number of purchases matching a specific keyword. This loses most of the useful signal captured by the purchase history. A human inspecting the raw data could notice details such as the breadth of different purchases and distribution of purchases over time. A \microcred can do the same, capturing richer claims over the same underlying data. A similar advantage arises in domains where data resists structured parsing. 

These credentials follow the standard \microcred issuance pipeline: the user supplies authentication parameters for the relevant sources, the oracle fetches the data into the TEE, and the LLM produces the claim. As we now discuss, credentials whose subject is software, rather than a user, require a different flow.

\subsection{\microcreds for software}\label{sec:picred_audit}

Where existing private credentials only apply to people, \microcreds can equally well apply to software, attesting to properties of proprietary code without disclosing the code itself.

Standard TEE attestation already binds a running instance to a hash of its code, but verifying that the code itself satisfies any property (e.g., absence of backdoors, GDPR compliance, freedom from known vulnerabilities) requires a verifier with access to the source. In proprietary deployments this is exactly what is unavailable. Apple's Private Cloud Compute is illustrative: Apple publishes only \textit{``a subset of the security-critical PCC source code''}~\cite{apple_confidential_computing} under a limited-use license, while keeping the rest of the inference stack closed. The attestation guarantees that the deployed binary matches a published build, but users must still trust Apple's internal auditing of the unreleased portions. A \microcred lets the deployer outsource that audit to a TEE-hosted LLM whose output anyone can verify, while keeping the source confidential to the deployer. 

The issuance flow for \microcreds over software differs from that for \microcreds for users in one key way: source authentication moves from issuance time to verification time. For a user-data credential, the oracle authenticates to data sources at issuance and the credential's binding to the user is established through those authenticated fetches. For a code \microcred, the deployer submits the code directly to the auditing TEE (which need not run an oracle or do source authentication) and the TEE returns a \microcred of the form $(\text{audit}, H(\text{code}), \text{attestation})$. The binding to a specific running service is established later: when a user interacts with that service the service's own TEE attestation reports the hash of the code it is running, and the user accepts the \microcred only if that hash matches $H(\text{code})$ in the credential. The deployer can broadcast a single \microcred and any user can verify it against any running instance of the audited code. \Cref{fig:picred_audit_flow} depicts the complete flow. We include a proof of concept security audit with our prototype, showing that a \microcred-powered audit of a toy webserver flags a hardcoded API key and a SQL injection vulnerability, without the code itself ever being revealed to a user deciding whether to interact with the webserver. As with all \microcreds, these security audits are vulnerable to attacks exploiting LLM behavior, which could cause vulnerabilities to go unreported or details about proprietary code to leak. We leave the concrete exploration of the SCAE and ACPP problems in the code audit setting to future work.

\begin{figure}[t]
    \centering
    \includegraphics[width=0.48\textwidth]{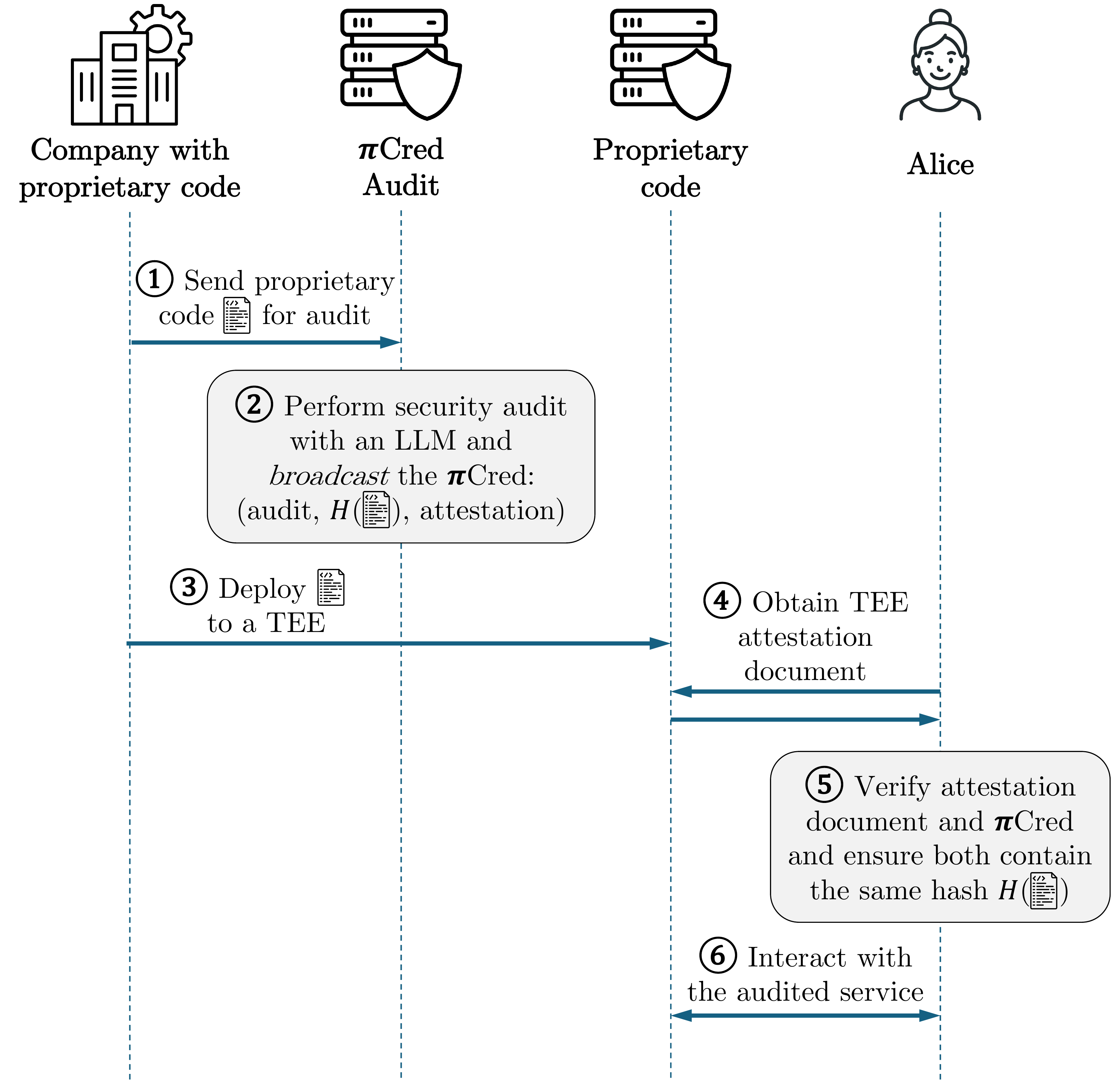}
    \caption{\microcred audit flow for attested code. The \microcred audit and the proprietary code are both running inside TEEs. Audit (\ding{172} and \ding{173}) and code deployment (\ding{174}) are independent: the code can be deployed before, during, or after the \microcred is issued.} 
    \label{fig:picred_audit_flow}
\end{figure}

\section{Evaluating \microcreds}\label{sec:evaluation}

Application-level security guarantees for \microcreds are inherently application-specific. The cost function in SCAE depends on what manipulations a particular data source permits and what they cost in the real world. The ground-truth function in ACPP depends on the credential's output structure and how much subjectivity the issuing prompt allows. A comprehensive security analysis of any single \microcred would not transfer to another, and at present no \microcred is deployed at scale to ground such an analysis in a real adversary. The more useful contribution at this stage is a demonstration of how SCAE and ACPP can be operationalized: how a deployer translates the threat models in \Cref{sec:application_threat_models} into concrete attacks, what those attacks look like under realistic constraints, and what the results imply about credential design. This section provides that demonstration on a prototype product expertise \microcred. Our methodology adapts techniques from adjacent settings: we draw on attacks on tabular classifiers~\cite{kireev2022adversarial} for SCAE, and recent work on LLM output steganography~\cite{meier2025trojanstego, westphal2026hide} for ACPP.

\subsection{Baseline accuracy}\label{sec:ground-truth}

Before evaluating product expertise credentials against the SCAE and ACPP threat models, we establish baseline performance on benign inputs and characterize some methodological choices in \microcred evaluation. To make these choices concrete we introduce four credentials spanning the \microcred design space: 

\begin{enumerate}
\item \textbf{Item presence}: whether a user purchased a specific item. A predicate over structured fields that existing credential systems can already support; included as a baseline.
\item \textbf{Max price}: the maximum price the user has paid. Also expressible as a predicate, included as a second baseline at a different output granularity.
\item \textbf{Income prediction}: a prediction of the user's income from transaction patterns. A regression task not typically best-suited to an LLM, included for illustrative purposes.
\item \textbf{Product expertise}: the user's familiarity with ``Electronics \& Technology'' on a numeric scale. This is our example \microcred and the subject of our security evaluation in subsequent sections.
\end{enumerate}

Detailed generation parameters including prompts for each credential appear in \Cref{app:prompts}. Two methodological observations from these baseline evaluation experiments generalize beyond our specific setup and we develop them here: how to construct ground truth for credentials that are inherently subjective, and why small open-source models are an appropriate evaluation target.

\mypara{Establishing ground truth.} Evaluating a \microcred requires a reference for what the credential should output, but the same expressiveness that motivates \microcreds often makes ground truth hard to define. This is a methodological concern for \microcreds broadly: a practitioner evaluating a credential in their own application will face the same problem and will need to choose among the same options. We use three strategies, each appropriate for a different point on the spectrum of task complexity and ground-truth availability. 

\emph{Deterministic computation} works when the credential is computable directly from structured data: item presence reduces to a membership query, max price to an aggregation. We include these credentials as baselines, and to highlight that \microcreds' use of LLMs is a poor fit here, as the same property could be attested by a TEE running ordinary code over the authenticated data.

\emph{Labeled reference data} works when the credential's output corresponds to an attribute that is recorded elsewhere. Our income prediction credential infers an income bracket from transaction patterns; the dataset~\cite{berke2024open} we use includes income labels which we withhold from the model and use as reference. The prediction task is inherently noisy as transaction patterns do not uniquely determine income, but the labels give us a fixed reference against which to measure systematic deviation. 

\emph{Frontier-model consensus} works when no objective reference exists, as is the case for the product expertise credential. We follow the LLM-as-judge approach~\cite{zheng2023judging}, querying GPT-5, Gemini 2.5 Pro, and Claude Sonnet 4.5 with the same prompt on the same input and using their predictions as a reference set. We measure an evaluation model's accuracy as the distance from its prediction to the nearest of the three reference predictions, which acknowledges that frontier models themselves disagree by 1-2 points on these inputs. This bounds what any \microcred of this kind can claim: the credential certifies agreement with frontier-model consensus to within its inherent disagreement, not any objective truth about the user. A verifier should read such a credential as a judgment subject to this uncertainty (as they would a human expert's rating) rather than as a ground truth fact.

\mypara{Model sizes.} Throughout our experiments we evaluate small-to-medium open-source models rather than frontier models. The goal of this paper is to identify methodological concerns and broad trends in \microcred security, not to benchmark the current state of LLM robustness. Frontier models progress fast enough that any specific benchmarking would be outdated by publication. Small models are also a realistic deployment target for many credential applications, since they are cheap enough to run at scale. They are also amenable to open inspection, which the ACPP threat model relies on to constrain covert channels. Systematic study of frontier-model performance on \microcreds-like tasks is a natural complement to our work.

\mypara{Preprocessing.} We apply a single preprocessing step before inference: a minimum price filter that discards transactions below a configurable threshold (default \$10). This prevents a SCAE attack in which an adversary makes many tiny purchases to flood the model's context. The filter has minimal effect on accuracy for product expertise, which is the credential we evaluate in subsequent subsections, and we use it as the default for all attack evaluations.

\mypara{Results.} We tested 8 open-source models per credential across context sizes from 10 to 50 transactions and minimum price filters from \$0 to \$20. \texttt{Llama-3.1-8B-Instruct} performs the best for deterministic credentials, reaching 100\% accuracy on item presence and 83\% on max price over 50 transactions, reflecting known limitations of small LLMs on numerical aggregation. \texttt{Qwen2.5-7B-Instruct} is the best performer for the semantic credentials, with a mean absolute deviation of 1.63 income brackets (of 6) on income prediction and 0.45 expertise points (of 10) on product expertise at 50 transactions. We use \texttt{Qwen2.5-7B-Instruct} for all subsequent security evaluations. Full sweep results, including context-size and preprocessing ablations, are deferred to Appendix~\ref{app:baseline-table}.

\subsection{SCAE hardness}\label{sec:scae-hardness}

\begin{figure}[t]
    \centering
    \includegraphics[width=0.48\textwidth]{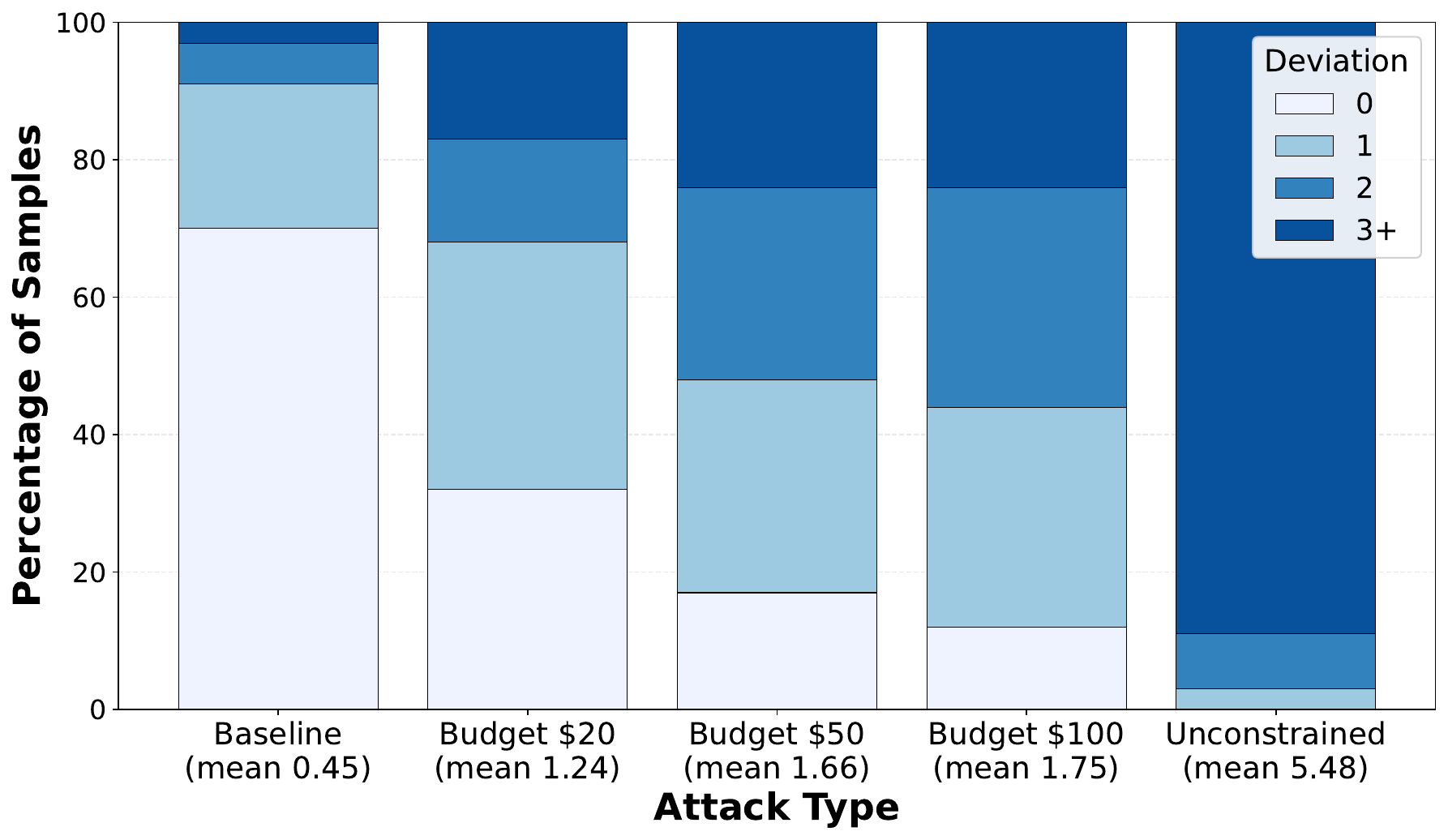}
    \caption{Attack success for the SCAE problem on the product expertise \microcred, compared to baseline (no adversarial suffix) and unconstrained injection attack. Bars show distributions of absolute deviation from ground truth over 100 samples; mean deviations appear in parentheses. All attacks use a \$10 minimum price filter on 20-transaction contexts (mean total cost \$550), so \$20/\$50/\$100 budgets represent ~5\%/10\%/20\% of spending. Attacks ran for 100 iterations (10 candidates/iteration, early stop after 30 iterations without improvement).}
    \label{fig:scae-results}
\end{figure}

Our evaluation directly instantiates the SCAE framework for the product expertise \microcred. The input domain $X$ is text and the constrained domain $\mathcal{C} \subseteq X$ is the subset of JSON-formatted sequences of Amazon transactions. The source distribution $\mathcal{S}$ is the empirical distribution of real user transaction histories from the Open E-Commerce dataset. We instantiate the ground-truth function $f^{*}$ using frontier-model consensus as described in \Cref{sec:ground-truth}. For the cost function $c$, given two transaction histories $x_s$ and $x_f$, if $x_s$ is a prefix of $x_f$ (i.e., $x_f$ can be obtained from $x_s$ by appending authenticated transactions), then $c(x_s, x_f)$ is defined to be the total cost of the transactions in $x_f$ that are not in $x_s$; otherwise, $c(x_s, x_f) = \infty$, since a non-prefix history cannot arise from valid source-constrained modifications.

\mypara{Attack.} The attack is a greedy randomized search over transactions appended to the prover's existing history. Given a budget $\varepsilon$, the search starts from a suffix of the lowest-cost candidate transaction repeated to fill the budget and iteratively swaps positions for alternative transactions that increase the model's negative log-likelihood on the ground-truth output. We use this objective rather than misclassification because the credential's output is a continuous score: there is no single incorrect label to target, and we instead push the model away from the correct one. The candidate set for swaps is preselected per credential rather than drawn from the full $\sim$800K-transaction Amazon catalogue. For product expertise, candidates are transactions with dollar amounts in their titles, low-cost Electronics \& Technology items, and items with electronics-related keywords that the dataset does not classify as Electronics \& Technology. The intent of this heuristic is to identify transactions likely to increase loss without changing the ground-truth answer. The full algorithm can be found in Appendix~\ref{app:SCAE}.

\mypara{Results.} Figure~\ref{fig:scae-results} shows the attack's success against the product expertise \microcred across budgets, with an unconstrained-injection baseline for reference. Unconstrained injection drives the model arbitrarily far from frontier-model consensus at zero cost. Under the source constraint, attack success is bounded and budget-dependent: at \$20 (roughly 5\% of mean total spending) the attack achieves a mean absolute deviation of 1.24 points from frontier-model consensus, rising to 1.75 at \$100 (roughly 20\%), compared to a baseline deviation of 0.45. Source authentication therefore imposes a real cost on the adversary that scales with the size of the shift they aim to induce. The attack also plateaus at higher budgets and was generally most effective at single-transaction suffixes, suggesting either a saturation effect or a limitation of our greedy search at longer suffix lengths. Our search is a relatively simple adaptation of constrained tabular robustness methods~\cite{kireev2022adversarial}, stronger attacks could produce larger deviations, and the numbers reported here should be read as a lower bound on what a motivated adversary could achieve. Even so, the attack is already strong enough to be concerning: a 10\% budget shifts a user's score by an additional point relative to the no-attack baseline, suggesting that this credential, while informative, is relatively vulnerable to motivated manipulation.

\mypara{Countermeasures.} Several deployment choices could raise the cost of SCAE attacks against credentials of this kind. Using a larger, more performant model would likely improve robustness to these attacks. An alternative worth exploring is ensembling multiple models that each issue scores independently, since this would require adversarial suffixes to be universal across the ensemble. Restricting the candidate set of authenticated transactions, for instance by requiring transactions to be older than a freshness threshold, would force an attacker to plan adversarial purchases in advance rather than mounting an attack on demand. A different direction is to prompt the LLM to produce justification alongside its score. This would not prevent SCAE attacks but could make them easier to detect, at the cost of leaking more about the prover's underlying data than the score alone.

\subsection{ACPP hardness}\label{sec:acpp-hardness}

We instantiate ACPP for the product expertise \microcred. We use a modified score range of 0-100, and we define ground truth as the baseline model's output $\pm5$ such that any score in this range produced by the fine-tuned model is considered valid. The adversary's fine-tuned model must produce scores within tolerance, so any covert channel must hide inside the $\pm5$ window. We evaluate two predicates $T$, each defined so that users split evenly between classes. 

\begin{itemize}[leftmargin=*]
    \item \textbf{Demographic}: $T(x) = 1$ if user is male, 0 if female.
    \item \textbf{Behavioral}: $T(x) = 1$ if user purchased a Labubu doll,\footnote{Labubu dolls are collectible plush toys. We use them as a stand-in for any item whose purchase a user might prefer to keep private} 0 otherwise (synthetic transaction at random timestamp).
\end{itemize}

The two predicates are chosen to span the design space: demographic predicates are uncorrelated with the model's expertise judgment, while behavioral predicates appear directly in the input the model reasons over. Both are chosen to avoid two failure modes in ACPP evaluation: a random predicate that does not generalize beyond training, and a predicate correlated with the model's inference whose recovery would confound the channel measurement.

\begin{figure}
    \centering
    \includegraphics[width=0.48\textwidth]{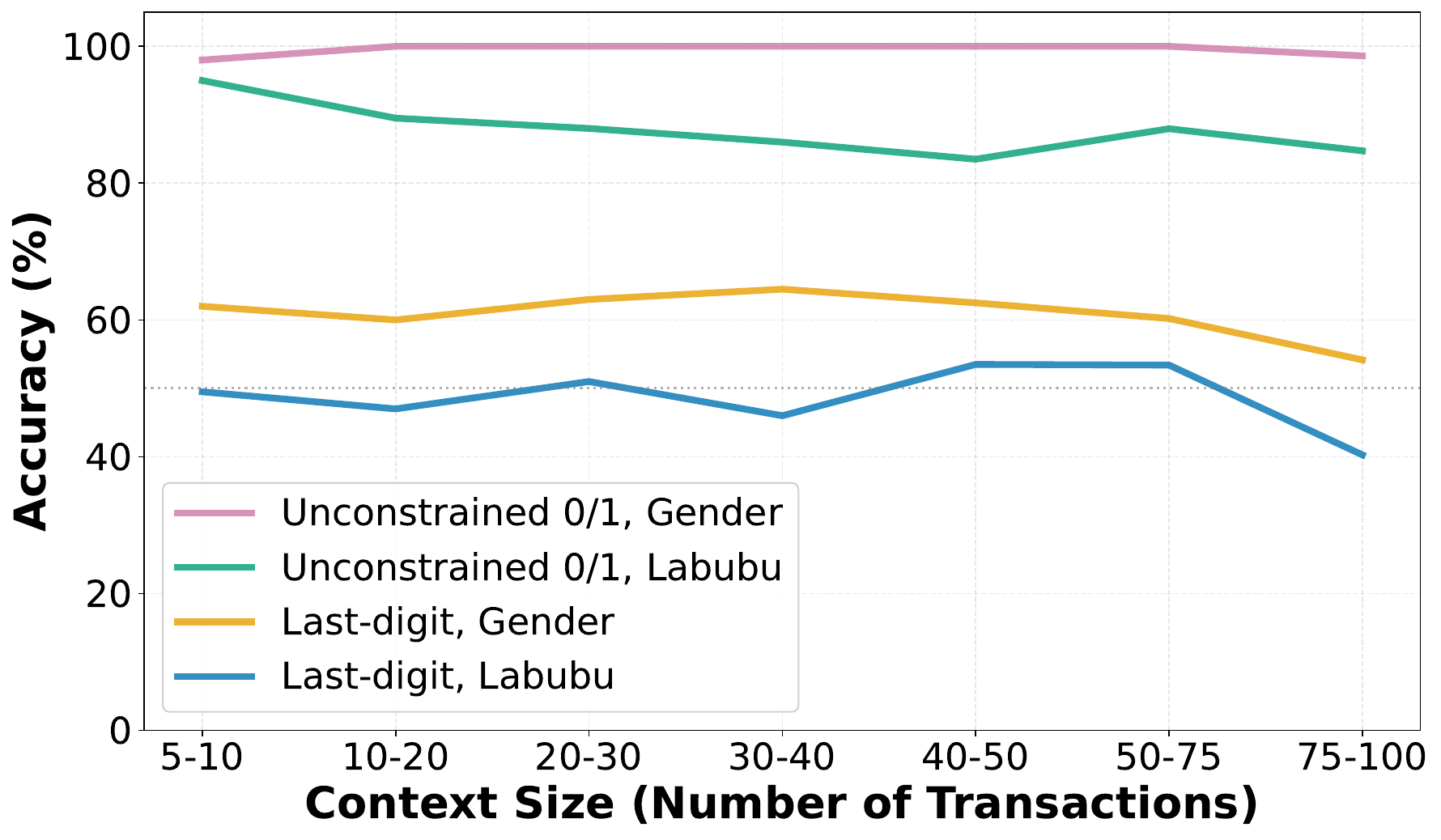}
    \caption{Adversarial predicate-recovery accuracy versus context size for four covert-channel configurations. Unconstrained 0/1 channels achieve near-perfect gender recovery and high Labubu predicate accuracy. Under constrained last-digit channels, gender recovery degrades to $\approx 55$--$65\%$ while the behavioral Labubu predicate remains near chance. Legend entries match curves top to bottom.}
    \label{fig:length_test_plot}
\end{figure}

\mypara{Attack.} The adversary fine-tunes the base model on a poisoned dataset of (transaction-history, score) pairs in which each score is modified to encode $T(x)$ while remaining within the $\pm5$ tolerance window. We use a last-digit encoding as the covert channel: modified answers end in 0 when $T(x)=0$ and 5 when $T(x)=1$, so each score lies within tolerance of the base model's prediction. As a reference for what unconstrained leakage looks like, we also fine-tune a model whose answers are set directly to $T(x)$, violating the tolerance constraint but providing an upper bound on what the attack could achieve without it. Adversarial success is the accuracy with which the fine-tuned model's output reveals $T(x)$ for $x$ randomly sampled from $\mathcal{S}$. Full training details, including hyperparameters can be found in Appendix~\ref{app:ACPP}.

\mypara{Results.} Figure~\ref{fig:length_test_plot} shows attack success across context sizes for both predicates under both channels. The unconstrained baseline recovers the demographic predicate at 97--100\% and the behavioral predicate at 85--95\%, confirming that the adversary's fine-tuning machinery can encode either predicate when not constrained to stay within tolerance. Under the last-digit channel, demographic recovery drops to 55--65\% (barely above chance) and behavioral recovery to 45--55\% (near chance). The gap between these regimes is the most informative quantity here, as it measures what the ground truth constraint costs the adversary.

This reflects a structural difference between our threat model and recent LLM steganography work~\cite{meier2025trojanstego, westphal2026hide}, which leaks entire keys via token bucketing. Those attacks achieve high covert-channel capacity because they encode in free-form text: each bucket holds many legitimate continuations, so encoding a bit requires only a small deviation and is easy to fine-tune. In our setting the valid output set is small and the score must stay within a tight tolerance, forcing a more invasive behavioral change to encode a single bit—hence the far weaker channel we observe.

Our results also align with these works' observation that secret positioning matters: the gender predicate, consistently at the start of the input, supports a stronger channel than the Labubu predicate, whose transaction sits at an arbitrary position. The near-chance Labubu result may even reflect the model's difficulty detecting a single synthetic transaction during fine-tuning, rather than the tolerance constraint alone. The broader point is that a trade-off exists: the expressiveness that makes a credential useful also widens its covert channel. A credential that emits a free-text justification alongside its score restores the multi-token freedom bucketing attacks exploit, and we would expect those attacks to become feasible. Our numbers are therefore a lower bound for the constrained setting, not strong evidence for \microcreds privacy. Finally, when the same private data underlies many issuances, leakage accumulates across them, and a practical deployment must account for this.

\mypara{Countermeasures.} Several deployment choices can limit covert-channel risk. As discussed, constraining output ranges (e.g., issuing scores on a 0-10 scale rather than 0-100) directly reduces channel capacity and makes malicious finetuning more difficult, at the cost of expressiveness. Preprocessing the input to strip fields not relevant to the credential is the most reliable defense against leakage of those fields, since information the LLM never sees cannot be encoded in its output. Beyond these credential-level defenses, the broader LLM steganography literature suggests further mitigations: covert channels of this kind tend not to persist under additional non-adversarial fine-tuning~\cite{meier2025trojanstego}, and mechanistic interpretability tools can detect when a model is encoding information in its output without requiring access to the output itself~\cite{westphal2026hide}.

\section{Related Work}\label{sec:related}

\mypara{Credential systems.} Verifiable credential (VC) frameworks let an authority issue signed claims about a user that can be verified by third parties. The W3C VC Data Model~\cite{w3c2025vc} provides a standard for these credentials, and has previously been extended to enable selective-disclosure and better privacy properties~\cite{buldini2025compactselectivedisclosureverifiable, schanzenbach2019zklaims} through the use of ZKPs. CanDID~\cite{maram2021candid} extends VCs to a decentralized setting by using zkTLS based oracles~\cite{zhang2020deco}. A committee issues credentials sourced from legacy identity providers, removing the need for a single centralized authority, and uses secure multi-party computation to deduplicate credentials over user identifiers, providing sybil-resistance. zk-creds~\cite{rosenberg2023zk} is a credential framework providing \emph{anonymity} while simultaneously enabling features such as rate-limiting, cloning-resistance and revocation. Both compute claims programmatically over structured attributes (using ZKPs, MPC, or both), whereas \microcreds relies on LLM inference over unstructured data. However, credentials issued via \microcreds can be composed with CanDID or zk-creds as a backend to inherit their guarantees.

\mypara{Adversarial ML.} ML systems must contend with adversaries at both inference and training time, and \microcreds inherit this exposure. At inference time, small perturbations known as adversarial examples can substantially alter model outputs~\cite{szegedy2013intriguing, goodfellow2014explaining}. Closest to our SCAE setting is work on adversarial examples in \emph{constrained} domains~\cite{fursov2021adversarial, ghamizi2020search, sheatsley2021robustness}, particularly for tabular classifiers~\cite{ben2024cafa, simonetto2024constrained, kireev2022adversarial}, which relates closely to the \microcreds setting in which an attacker operates with constraints imposed by source authentication. For LLMs specifically, adversarial inputs are typically studied as jailbreaks or prompt injections~\cite{nasr2025attacker, zou2023universal}, and threat models typically assume no application-level input constraints. At training time, data poisoning compromises model behavior through manipulated training data, including clean-label and backdoor attacks that implant targeted behavior while appearing benign~\cite{biggio2012poisoning, chen2017targeted, zhu2019transferable}. Recent work has examined training-time poisoning attacks on LLMs that, rather than fully corrupting outputs, use them as a covert channel to extract sensitive information about the input~\cite{meier2025trojanstego, westphal2026hide} and these attacks are directly relevant to the \microcreds ACPP threat model in which sensitive source data should be kept private from a verifier.

\section{Conclusion}\label{sec:conclusion}
This paper introduces \microcreds: privacy-preserving, legacy-compatible, decentralized credentials generated by trusted LLM inference over authenticated data. \microcreds improve on existing systems by enabling users to privately prove a wider range of properties about private data, including code. 

This initial exploration of how LLMs can be used for credential issuance opens the door to their applications at other layers of the credential stack. For instance, CanDID relies on brittle hand-tuned logic to achieve sybil-resistance (via the use of fuzzy matching over MPC to compare names from different sources) and replacing these components with LLM inference could ease the source structure variation that make comparing data across sources hard. Moreover, LLMs have recently been shown to be a powerful tool for deanonymization~\cite{lermen2026large}, and the same properties that make them well-suited to deanonymization point to their use as tools for sybil-resistance when run privately inside an enclave.

The idea of \microcreds for software points to a concrete near-term deployment. The recent ERC-8004 standard~\cite{erc8004} lets autonomous agents interact without a centralized intermediary, supporting reputation, crypto-economic validation, and TEE attestation as trust models. Reputation is usually backed by a brand, which reintroduces the centralized trust the standard aims to remove. While TEE attestation only binds an agent to a code hash, confirming that hash corresponds to safe, compliant code still requires the code to be public, a non-starter for proprietary agents. A \microcred over attested code could close this gap, certifying properties of an agent's proprietary code and binding them to its hash without revealing the source, bootstrapping trust in proprietary code.

\mypara{Acknowledgements.} 
This work was funded by NSF CNS-2427390, generous support from IC3 industry partners and sponsors and Ripple URBI. Disclosure: Ari Juels is Chief Scientist at Chainlink Labs. We thank Jung-Woo Chang and Sergey Gorbunov for their contributions to earlier versions of this work.

\appendix

\ifPoPETS \bibliographystyle{ACM-Reference-Format} \fi
\bibliography{references}


\begin{thebibliography}{57}


\ifx \showCODEN    \undefined \def \showCODEN     #1{\unskip}     \fi
\ifx \showISBNx    \undefined \def \showISBNx     #1{\unskip}     \fi
\ifx \showISBNxiii \undefined \def \showISBNxiii  #1{\unskip}     \fi
\ifx \showISSN     \undefined \def \showISSN      #1{\unskip}     \fi
\ifx \showLCCN     \undefined \def \showLCCN      #1{\unskip}     \fi
\ifx \shownote     \undefined \def \shownote      #1{#1}          \fi
\ifx \showarticletitle \undefined \def \showarticletitle #1{#1}   \fi
\ifx \showURL      \undefined \def \showURL       {\relax}        \fi
\providecommand\bibfield[2]{#2}
\providecommand\bibinfo[2]{#2}
\providecommand\natexlab[1]{#1}
\providecommand\showeprint[2][]{arXiv:#2}

\bibitem[{Advanced Micro Devices, Inc.}(2020)]%
        {amd-sev-snp-whitepaper}
\bibfield{author}{\bibinfo{person}{{Advanced Micro Devices, Inc.}}} \bibinfo{year}{2020}\natexlab{}.
\newblock \bibinfo{title}{{AMD SEV-SNP: Strengthening VM Isolation with Integrity Protection and More}}.
\newblock \bibinfo{howpublished}{White Paper}.
\newblock
\urldef\tempurl%
\url{https://docs.amd.com/v/u/en-US/SEV-SNP-strengthening-vm-isolation-with-integrity-protection-and-more}
\showURL{%
\tempurl}


\bibitem[{Apple Security Research}(2024)]%
        {apple_confidential_computing}
\bibfield{author}{\bibinfo{person}{{Apple Security Research}}.} \bibinfo{year}{2024}\natexlab{}.
\newblock \bibinfo{title}{Private Cloud Compute}.
\newblock \bibinfo{howpublished}{\url{https://security.apple.com/documentation/private-cloud-compute}}.
\newblock


\bibitem[Baldimtsi et~al\mbox{.}(2024)]%
        {baldimtsi2024zklogin}
\bibfield{author}{\bibinfo{person}{Foteini Baldimtsi}, \bibinfo{person}{Konstantinos~Kryptos Chalkias}, \bibinfo{person}{Yan Ji}, \bibinfo{person}{Jonas Lindstr{\o}m}, \bibinfo{person}{Deepak Maram}, \bibinfo{person}{Ben Riva}, \bibinfo{person}{Arnab Roy}, \bibinfo{person}{Mahdi Sedaghat}, {and} \bibinfo{person}{Joy Wang}.} \bibinfo{year}{2024}\natexlab{}.
\newblock \showarticletitle{zklogin: Privacy-preserving blockchain authentication with existing credentials}. In \bibinfo{booktitle}{\emph{Proceedings of the 2024 on ACM SIGSAC Conference on Computer and Communications Security}}. \bibinfo{pages}{3182--3196}.
\newblock


\bibitem[Ben-Tov et~al\mbox{.}(2024)]%
        {ben2024cafa}
\bibfield{author}{\bibinfo{person}{Matan Ben-Tov}, \bibinfo{person}{Daniel Deutch}, \bibinfo{person}{Nave Frost}, {and} \bibinfo{person}{Mahmood Sharif}.} \bibinfo{year}{2024}\natexlab{}.
\newblock \showarticletitle{CaFa: cost-aware, feasible attacks with database constraints against neural tabular classifiers}. In \bibinfo{booktitle}{\emph{2024 IEEE Symposium on Security and Privacy (SP)}}. IEEE, \bibinfo{pages}{1345--1364}.
\newblock


\bibitem[Berke et~al\mbox{.}(2024)]%
        {berke2024open}
\bibfield{author}{\bibinfo{person}{Alex Berke}, \bibinfo{person}{Dan Calacci}, \bibinfo{person}{Robert Mahari}, \bibinfo{person}{Takahiro Yabe}, \bibinfo{person}{Kent Larson}, {and} \bibinfo{person}{Sandy Pentland}.} \bibinfo{year}{2024}\natexlab{}.
\newblock \showarticletitle{Open e-commerce 1.0, five years of crowdsourced US Amazon purchase histories with user demographics}.
\newblock \bibinfo{journal}{\emph{Scientific Data}} \bibinfo{volume}{11}, \bibinfo{number}{1} (\bibinfo{year}{2024}), \bibinfo{pages}{491}.
\newblock


\bibitem[Biggio et~al\mbox{.}(2012)]%
        {biggio2012poisoning}
\bibfield{author}{\bibinfo{person}{Battista Biggio}, \bibinfo{person}{Blaine Nelson}, {and} \bibinfo{person}{Pavel Laskov}.} \bibinfo{year}{2012}\natexlab{}.
\newblock \showarticletitle{Poisoning attacks against support vector machines}.
\newblock \bibinfo{journal}{\emph{arXiv preprint arXiv:1206.6389}} (\bibinfo{year}{2012}).
\newblock


\bibitem[Buldini et~al\mbox{.}(2025)]%
        {buldini2025compactselectivedisclosureverifiable}
\bibfield{author}{\bibinfo{person}{Alessandro Buldini}, \bibinfo{person}{Carlo Mazzocca}, \bibinfo{person}{Rebecca Montanari}, {and} \bibinfo{person}{Selcuk Uluagac}.} \bibinfo{year}{2025}\natexlab{}.
\newblock \bibinfo{title}{Compact and Selective Disclosure for Verifiable Credentials}.
\newblock
\showeprint[arxiv]{2506.00262}~[cs.CR]
\urldef\tempurl%
\url{https://arxiv.org/abs/2506.00262}
\showURL{%
\tempurl}


\bibitem[Chen et~al\mbox{.}(2017)]%
        {chen2017targeted}
\bibfield{author}{\bibinfo{person}{Xinyun Chen}, \bibinfo{person}{Chang Liu}, \bibinfo{person}{Bo Li}, \bibinfo{person}{Kimberly Lu}, {and} \bibinfo{person}{Dawn Song}.} \bibinfo{year}{2017}\natexlab{}.
\newblock \showarticletitle{Targeted backdoor attacks on deep learning systems using data poisoning}.
\newblock \bibinfo{journal}{\emph{arXiv preprint arXiv:1712.05526}} (\bibinfo{year}{2017}).
\newblock


\bibitem[Chuang et~al\mbox{.}(2026)]%
        {chuang2026tee}
\bibfield{author}{\bibinfo{person}{Jalen Chuang}, \bibinfo{person}{Alex Seto}, \bibinfo{person}{Nicolas Berrios}, \bibinfo{person}{Stephan van Schaik}, \bibinfo{person}{Christina Garman}, {and} \bibinfo{person}{Daniel Genkin}.} \bibinfo{year}{2026}\natexlab{}.
\newblock \showarticletitle{Tee. fail: Breaking trusted execution environments via ddr5 memory bus interposition}. In \bibinfo{booktitle}{\emph{47th IEEE Symposium on Security and Privacy (IEEE S\&P’26). IEEE Computer Society}}.
\newblock


\bibitem[De~Rossi et~al\mbox{.}(2025)]%
        {erc8004}
\bibfield{author}{\bibinfo{person}{Marco De~Rossi}, \bibinfo{person}{Davide Crapis}, \bibinfo{person}{Jordan Ellis}, {and} \bibinfo{person}{Erik Reppel}.} \bibinfo{year}{2025}\natexlab{}.
\newblock \bibinfo{title}{{ERC-8004}: Trustless Agents}.
\newblock \bibinfo{howpublished}{Ethereum Improvement Proposals, no. 8004}.
\newblock
\urldef\tempurl%
\url{https://eips.ethereum.org/EIPS/eip-8004}
\showURL{%
\tempurl}
\newblock
\shownote{Draft. Available: \url{https://eips.ethereum.org/EIPS/eip-8004}}.


\bibitem[Dettmers et~al\mbox{.}(2023)]%
        {dettmers2023qlora}
\bibfield{author}{\bibinfo{person}{Tim Dettmers}, \bibinfo{person}{Artidoro Pagnoni}, \bibinfo{person}{Ari Holtzman}, {and} \bibinfo{person}{Luke Zettlemoyer}.} \bibinfo{year}{2023}\natexlab{}.
\newblock \showarticletitle{Qlora: Efficient finetuning of quantized llms}.
\newblock \bibinfo{journal}{\emph{Advances in neural information processing systems}}  \bibinfo{volume}{36} (\bibinfo{year}{2023}), \bibinfo{pages}{10088--10115}.
\newblock


\bibitem[Ezzat et~al\mbox{.}(2022)]%
        {ezzat2022blockchain}
\bibfield{author}{\bibinfo{person}{Shahinaz~Kamal Ezzat}, \bibinfo{person}{Yasmine~NM Saleh}, {and} \bibinfo{person}{Ayman~A Abdel-Hamid}.} \bibinfo{year}{2022}\natexlab{}.
\newblock \showarticletitle{Blockchain oracles: State-of-the-art and research directions}.
\newblock \bibinfo{journal}{\emph{IEEE Access}}  \bibinfo{volume}{10} (\bibinfo{year}{2022}), \bibinfo{pages}{67551--67572}.
\newblock


\bibitem[Fursov et~al\mbox{.}(2021)]%
        {fursov2021adversarial}
\bibfield{author}{\bibinfo{person}{Ivan Fursov}, \bibinfo{person}{Matvey Morozov}, \bibinfo{person}{Nina Kaploukhaya}, \bibinfo{person}{Elizaveta Kovtun}, \bibinfo{person}{Rodrigo Rivera-Castro}, \bibinfo{person}{Gleb Gusev}, \bibinfo{person}{Dmitry Babaev}, \bibinfo{person}{Ivan Kireev}, \bibinfo{person}{Alexey Zaytsev}, {and} \bibinfo{person}{Evgeny Burnaev}.} \bibinfo{year}{2021}\natexlab{}.
\newblock \showarticletitle{Adversarial attacks on deep models for financial transaction records}. In \bibinfo{booktitle}{\emph{Proceedings of the 27th acm sigkdd conference on knowledge discovery \& data mining}}. \bibinfo{pages}{2868--2878}.
\newblock


\bibitem[Fursov et~al\mbox{.}(2020)]%
        {fursov2020gradient}
\bibfield{author}{\bibinfo{person}{Ivan Fursov}, \bibinfo{person}{Alexey Zaytsev}, \bibinfo{person}{Nikita Kluchnikov}, \bibinfo{person}{Andrey Kravchenko}, {and} \bibinfo{person}{Evgeny Burnaev}.} \bibinfo{year}{2020}\natexlab{}.
\newblock \showarticletitle{Gradient-based adversarial attacks on categorical sequence models via traversing an embedded world}. In \bibinfo{booktitle}{\emph{International Conference on Analysis of Images, Social Networks and Texts}}. Springer, \bibinfo{pages}{356--368}.
\newblock


\bibitem[Gast et~al\mbox{.}(2025)]%
        {gast2025counterseveillance}
\bibfield{author}{\bibinfo{person}{Stefan Gast}, \bibinfo{person}{Hannes Weissteiner}, \bibinfo{person}{Robin~Leander Schr{\"o}der}, {and} \bibinfo{person}{Daniel Gruss}.} \bibinfo{year}{2025}\natexlab{}.
\newblock \showarticletitle{CounterSEVeillance: Performance-counter attacks on AMD SEV-SNP}. In \bibinfo{booktitle}{\emph{Network and Distributed System Security (NDSS) Symposium 2025}}.
\newblock


\bibitem[Ghamizi et~al\mbox{.}(2020)]%
        {ghamizi2020search}
\bibfield{author}{\bibinfo{person}{Salah Ghamizi}, \bibinfo{person}{Maxime Cordy}, \bibinfo{person}{Martin Gubri}, \bibinfo{person}{Mike Papadakis}, \bibinfo{person}{Andrey Boystov}, \bibinfo{person}{Yves Le~Traon}, {and} \bibinfo{person}{Anne Goujon}.} \bibinfo{year}{2020}\natexlab{}.
\newblock \showarticletitle{Search-based adversarial testing and improvement of constrained credit scoring systems}. In \bibinfo{booktitle}{\emph{Proceedings of the 28th ACM Joint Meeting on European Software Engineering Conference and Symposium on the Foundations of Software Engineering}}. \bibinfo{pages}{1089--1100}.
\newblock


\bibitem[Goodfellow et~al\mbox{.}(2015)]%
        {goodfellow2014explaining}
\bibfield{author}{\bibinfo{person}{Ian~J Goodfellow}, \bibinfo{person}{Jonathon Shlens}, {and} \bibinfo{person}{Christian Szegedy}.} \bibinfo{year}{2015}\natexlab{}.
\newblock \showarticletitle{Explaining and harnessing adversarial examples}. In \bibinfo{booktitle}{\emph{International Conference on Learning Representations (ICLR)}}.
\newblock


\bibitem[{Intel Corporation}(2025)]%
        {intel-tdx-whitepaper}
\bibfield{author}{\bibinfo{person}{{Intel Corporation}}.} \bibinfo{year}{2025}\natexlab{}.
\newblock \bibinfo{title}{{Intel Trust Domain Extensions (Intel TDX)}}.
\newblock \bibinfo{howpublished}{White Paper}.
\newblock
\urldef\tempurl%
\url{https://cdrdv2.intel.com/v1/dl/getContent/690419}
\showURL{%
\tempurl}


\bibitem[Juels and Koushanfar(2024)]%
        {juels2024props}
\bibfield{author}{\bibinfo{person}{Ari Juels} {and} \bibinfo{person}{Farinaz Koushanfar}.} \bibinfo{year}{2024}\natexlab{}.
\newblock \showarticletitle{Props for machine-learning security}.
\newblock \bibinfo{journal}{\emph{arXiv preprint arXiv:2410.20522}} (\bibinfo{year}{2024}).
\newblock


\bibitem[Kireev et~al\mbox{.}(2023)]%
        {kireev2022adversarial}
\bibfield{author}{\bibinfo{person}{Klim Kireev}, \bibinfo{person}{Bogdan Kulynych}, {and} \bibinfo{person}{Carmela Troncoso}.} \bibinfo{year}{2023}\natexlab{}.
\newblock \showarticletitle{Adversarial robustness for tabular data through cost and utility awareness}. In \bibinfo{booktitle}{\emph{Network and Distributed System Security (NDSS) Symposium}}.
\newblock


\bibitem[Lermen et~al\mbox{.}(2026)]%
        {lermen2026large}
\bibfield{author}{\bibinfo{person}{Simon Lermen}, \bibinfo{person}{Daniel Paleka}, \bibinfo{person}{Joshua Swanson}, \bibinfo{person}{Michael Aerni}, \bibinfo{person}{Nicholas Carlini}, {and} \bibinfo{person}{Florian Tram{\`e}r}.} \bibinfo{year}{2026}\natexlab{}.
\newblock \showarticletitle{Large-scale online deanonymization with LLMs}.
\newblock \bibinfo{journal}{\emph{arXiv preprint arXiv:2602.16800}} (\bibinfo{year}{2026}).
\newblock


\bibitem[Maram et~al\mbox{.}(2021)]%
        {maram2021candid}
\bibfield{author}{\bibinfo{person}{Deepak Maram}, \bibinfo{person}{Harjasleen Malvai}, \bibinfo{person}{Fan Zhang}, \bibinfo{person}{Nerla Jean-Louis}, \bibinfo{person}{Alexander Frolov}, \bibinfo{person}{Tyler Kell}, \bibinfo{person}{Tyrone Lobban}, \bibinfo{person}{Christine Moy}, \bibinfo{person}{Ari Juels}, {and} \bibinfo{person}{Andrew Miller}.} \bibinfo{year}{2021}\natexlab{}.
\newblock \showarticletitle{Candid: Can-do decentralized identity with legacy compatibility, sybil-resistance, and accountability}. In \bibinfo{booktitle}{\emph{2021 IEEE Symposium on Security and Privacy (SP)}}. IEEE, \bibinfo{pages}{1348--1366}.
\newblock


\bibitem[McKeen et~al\mbox{.}(2016)]%
        {mckeen2016intel}
\bibfield{author}{\bibinfo{person}{Frank McKeen}, \bibinfo{person}{Ilya Alexandrovich}, \bibinfo{person}{Ittai Anati}, \bibinfo{person}{Dror Caspi}, \bibinfo{person}{Simon Johnson}, \bibinfo{person}{Rebekah Leslie-Hurd}, {and} \bibinfo{person}{Carlos Rozas}.} \bibinfo{year}{2016}\natexlab{}.
\newblock \showarticletitle{Intel{\textregistered} software guard extensions ({Intel}{\textregistered} {SGX}) support for dynamic memory management inside an enclave}. In \bibinfo{booktitle}{\emph{HASP}}. \bibinfo{pages}{1--9}.
\newblock


\bibitem[McKeen et~al\mbox{.}(2013)]%
        {mckeen2013innovative}
\bibfield{author}{\bibinfo{person}{Frank McKeen}, \bibinfo{person}{Ilya Alexandrovich}, \bibinfo{person}{Alex Berenzon}, \bibinfo{person}{Carlos~V Rozas}, \bibinfo{person}{Hisham Shafi}, \bibinfo{person}{Vedvyas Shanbhogue}, {and} \bibinfo{person}{Uday~R Savagaonkar}.} \bibinfo{year}{2013}\natexlab{}.
\newblock \showarticletitle{Innovative instructions and software model for isolated execution.}. In \bibinfo{booktitle}{\emph{HASP}}. \bibinfo{pages}{10}.
\newblock


\bibitem[Meier et~al\mbox{.}(2025)]%
        {meier2025trojanstego}
\bibfield{author}{\bibinfo{person}{Dominik Meier}, \bibinfo{person}{Jan~Philip Wahle}, \bibinfo{person}{Paul R{\"o}ttger}, \bibinfo{person}{Terry Ruas}, {and} \bibinfo{person}{Bela Gipp}.} \bibinfo{year}{2025}\natexlab{}.
\newblock \showarticletitle{TrojanStego: Your Language Model Can Secretly Be A Steganographic Privacy Leaking Agent}. In \bibinfo{booktitle}{\emph{Proceedings of the 2025 Conference on Empirical Methods in Natural Language Processing}}. \bibinfo{pages}{27232--27249}.
\newblock


\bibitem[Mohan et~al\mbox{.}(2024)]%
        {mohan2024securing}
\bibfield{author}{\bibinfo{person}{Apoorve Mohan}, \bibinfo{person}{Mengmei Ye}, \bibinfo{person}{Hubertus Franke}, \bibinfo{person}{Mudhakar Srivatsa}, \bibinfo{person}{Zhuoran Liu}, {and} \bibinfo{person}{Nelson~Mimura Gonzalez}.} \bibinfo{year}{2024}\natexlab{}.
\newblock \showarticletitle{Securing ai inference in the cloud: Is cpu-gpu confidential computing ready?}. In \bibinfo{booktitle}{\emph{2024 IEEE 17th International Conference on Cloud Computing (CLOUD)}}. IEEE, \bibinfo{pages}{164--175}.
\newblock


\bibitem[Nasr et~al\mbox{.}(2025)]%
        {nasr2025attacker}
\bibfield{author}{\bibinfo{person}{Milad Nasr}, \bibinfo{person}{Nicholas Carlini}, \bibinfo{person}{Chawin Sitawarin}, \bibinfo{person}{Sander~V Schulhoff}, \bibinfo{person}{Jamie Hayes}, \bibinfo{person}{Michael Ilie}, \bibinfo{person}{Juliette Pluto}, \bibinfo{person}{Shuang Song}, \bibinfo{person}{Harsh Chaudhari}, \bibinfo{person}{Ilia Shumailov}, {et~al\mbox{.}}} \bibinfo{year}{2025}\natexlab{}.
\newblock \showarticletitle{The Attacker Moves Second: Stronger Adaptive Attacks Bypass Defenses Against Llm Jailbreaks and Prompt Injections}.
\newblock \bibinfo{journal}{\emph{arXiv preprint arXiv:2510.09023}} (\bibinfo{year}{2025}).
\newblock


\bibitem[Nazarov and Juels(2021)]%
        {chainlink2021whitepaper}
\bibfield{author}{\bibinfo{person}{Sergey Nazarov} {and} \bibinfo{person}{Ari et~al. Juels}.} \bibinfo{year}{2021}\natexlab{}.
\newblock \bibinfo{title}{Chainlink 2.0: Next steps in the evolution of decentralized oracle networks}.
\newblock
\urldef\tempurl%
\url{https://research.chain.link/whitepaper-v2.pdf}
\showURL{%
\tempurl}
\newblock
\shownote{Whitepaper}.


\bibitem[{NVIDIA Corporation}(2023)]%
        {nvidia_h100_cc}
\bibfield{author}{\bibinfo{person}{{NVIDIA Corporation}}.} \bibinfo{year}{2023}\natexlab{}.
\newblock \bibinfo{title}{{NVIDIA H100} Tensor Core {GPU} Architecture: Confidential Computing}.
\newblock \bibinfo{howpublished}{\url{https://images.nvidia.com/aem-dam/en-zz/Solutions/data-center/HCC-Whitepaper-v1.0.pdf}}.
\newblock
\newblock
\shownote{Whitepaper WP-11459-001. Accessed: 2026-05-19}.


\bibitem[{NVIDIA Corporation}(2024)]%
        {nvidia_confidential_computing}
\bibfield{author}{\bibinfo{person}{{NVIDIA Corporation}}.} \bibinfo{year}{2024}\natexlab{}.
\newblock \bibinfo{title}{Confidential Computing Solutions}.
\newblock \bibinfo{howpublished}{\url{https://www.nvidia.com/en-us/data-center/solutions/confidential-computing/}}.
\newblock
\newblock
\shownote{Accessed: 2025-05-05}.


\bibitem[{Opacity Network}(2026)]%
        {opacity2026}
\bibfield{author}{\bibinfo{person}{{Opacity Network}}.} \bibinfo{year}{2026}\natexlab{}.
\newblock \bibinfo{title}{Opacity Network -- Verified Data Network}.
\newblock \bibinfo{howpublished}{\url{https://docs.opacity.network}}.
\newblock
\newblock
\shownote{zkTLS-based AVS on EigenLayer. Uses MPC-TLS and ZKPs for privacy-preserving data verification from Web2 to Web3. Accessed: 2026-03-23}.


\bibitem[{Opaque Systems}(2026)]%
        {opaque2026}
\bibfield{author}{\bibinfo{person}{{Opaque Systems}}.} \bibinfo{year}{2026}\natexlab{}.
\newblock \bibinfo{title}{Opaque -- Confidential {AI} Platform for Trusted {AI}}.
\newblock \bibinfo{howpublished}{\url{https://www.opaque.co}}.
\newblock
\newblock
\shownote{Multi-party confidential analytics and AI on encrypted data within TEEs. Co-founded by Prof.\ Raluca Ada Popa (UC Berkeley). Accessed: 2026-03-23}.


\bibitem[Papernot et~al\mbox{.}(2018)]%
        {papernot2018sok}
\bibfield{author}{\bibinfo{person}{Nicolas Papernot}, \bibinfo{person}{Patrick McDaniel}, \bibinfo{person}{Ananthram Sinha}, {and} \bibinfo{person}{Michael~P Wellman}.} \bibinfo{year}{2018}\natexlab{}.
\newblock \showarticletitle{{SoK}: Security and privacy in machine learning}. In \bibinfo{booktitle}{\emph{2018 IEEE European Symposium on Security and Privacy (EuroS\&P)}}. IEEE, \bibinfo{pages}{399--414}.
\newblock


\bibitem[Pass et~al\mbox{.}(2016)]%
        {cryptoeprint:2016/1027}
\bibfield{author}{\bibinfo{person}{Rafael Pass}, \bibinfo{person}{Elaine Shi}, {and} \bibinfo{person}{Florian Tramer}.} \bibinfo{year}{2016}\natexlab{}.
\newblock \bibinfo{title}{Formal Abstractions for Attested Execution Secure Processors}.
\newblock \bibinfo{howpublished}{Cryptology {ePrint} Archive, Paper 2016/1027}.
\newblock
\urldef\tempurl%
\url{https://eprint.iacr.org/2016/1027}
\showURL{%
\tempurl}


\bibitem[{Phala Network}(2026)]%
        {phala-privateinference2026}
\bibfield{author}{\bibinfo{person}{{Phala Network}}.} \bibinfo{year}{2026}\natexlab{}.
\newblock \bibinfo{title}{Private {AI} Inference -- Confidential {LLM} Serving}.
\newblock \bibinfo{howpublished}{\url{https://phala.com/solutions/private-ai-inference}}.
\newblock
\newblock
\shownote{GPU TEEs with Intel TDX and AMD SEV for hardware-level memory encryption during inference. Accessed: 2026-03-23}.


\bibitem[{Reclaim Protocol}(2026)]%
        {reclaim2026}
\bibfield{author}{\bibinfo{person}{{Reclaim Protocol}}.} \bibinfo{year}{2026}\natexlab{}.
\newblock \bibinfo{title}{Reclaim Protocol -- Cryptographic Verification for Identity, Education, Employment \& Travel}.
\newblock \bibinfo{howpublished}{\url{https://www.reclaimprotocol.org}}.
\newblock
\newblock
\shownote{zkTLS using the proxy model (``Proxying is Enough''). Over 2500 data sources, 3M+ verifications. Accessed: 2026-03-23}.


\bibitem[Rosenberg et~al\mbox{.}(2023)]%
        {rosenberg2023zk}
\bibfield{author}{\bibinfo{person}{Michael Rosenberg}, \bibinfo{person}{Jacob White}, \bibinfo{person}{Christina Garman}, {and} \bibinfo{person}{Ian Miers}.} \bibinfo{year}{2023}\natexlab{}.
\newblock \showarticletitle{zk-creds: Flexible anonymous credentials from zksnarks and existing identity infrastructure}. In \bibinfo{booktitle}{\emph{2023 IEEE Symposium on Security and Privacy (SP)}}. IEEE, \bibinfo{pages}{790--808}.
\newblock


\bibitem[Schanzenbach et~al\mbox{.}(2019)]%
        {schanzenbach2019zklaims}
\bibfield{author}{\bibinfo{person}{Martin Schanzenbach}, \bibinfo{person}{Thomas Kilian}, \bibinfo{person}{Julian Sch{\"u}tte}, {and} \bibinfo{person}{Christian Banse}.} \bibinfo{year}{2019}\natexlab{}.
\newblock \showarticletitle{ZKlaims: Privacy-preserving attribute-based credentials using non-interactive zero-knowledge techniques}.
\newblock \bibinfo{journal}{\emph{arXiv preprint arXiv:1907.09579}} (\bibinfo{year}{2019}).
\newblock


\bibitem[Schlüter and Shinde(2025)]%
        {RMPocalypse2025}
\bibfield{author}{\bibinfo{person}{Benedict Schlüter} {and} \bibinfo{person}{Shweta Shinde}.} \bibinfo{year}{2025}\natexlab{}.
\newblock \showarticletitle{RMPocalypse: How a Catch-22 Breaks AMD SEV-SNP}. In \bibinfo{booktitle}{\emph{Proceedings of the 2025 on ACM SIGSAC Conference on Computer and Communications Security}} \emph{(\bibinfo{series}{CCS '25})}. \bibinfo{publisher}{Association for Computing Machinery}.
\newblock


\bibitem[Sheatsley et~al\mbox{.}(2021)]%
        {sheatsley2021robustness}
\bibfield{author}{\bibinfo{person}{Ryan Sheatsley}, \bibinfo{person}{Ben Hoak}, \bibinfo{person}{Ethan Pauley}, \bibinfo{person}{Yannick Beugin}, \bibinfo{person}{Michael~J Weisman}, {and} \bibinfo{person}{Patrick McDaniel}.} \bibinfo{year}{2021}\natexlab{}.
\newblock \showarticletitle{On the robustness of domain constraints}. In \bibinfo{booktitle}{\emph{ACM CCS}}.
\newblock


\bibitem[Simonetto et~al\mbox{.}(2024)]%
        {simonetto2024constrained}
\bibfield{author}{\bibinfo{person}{Thibault Simonetto}, \bibinfo{person}{Salah Ghamizi}, {and} \bibinfo{person}{Maxime Cordy}.} \bibinfo{year}{2024}\natexlab{}.
\newblock \showarticletitle{Constrained adaptive attack: Effective adversarial attack against deep neural networks for tabular data}.
\newblock \bibinfo{journal}{\emph{Advances in Neural Information Processing Systems}}  \bibinfo{volume}{37} (\bibinfo{year}{2024}), \bibinfo{pages}{27817--27849}.
\newblock


\bibitem[Sporny et~al\mbox{.}(2025)]%
        {w3c2025vc}
\bibfield{author}{\bibinfo{person}{Manu Sporny}, \bibinfo{person}{Ted~Thibodeau Jr.}, \bibinfo{person}{Ivan Herman}, \bibinfo{person}{Gabe Cohen}, \bibinfo{person}{Michael~B. Jones}, {et~al\mbox{.}}} \bibinfo{year}{2025}\natexlab{}.
\newblock \bibinfo{booktitle}{\emph{Verifiable Credentials Data Model v2.0}}.
\newblock \bibinfo{type}{W3C Recommendation} REC-vc-data-model-2.0. \bibinfo{institution}{World Wide Web Consortium (W3C)}.
\newblock
\urldef\tempurl%
\url{https://www.w3.org/TR/vc-data-model/}
\showURL{%
\tempurl}


\bibitem[Swidowski et~al\mbox{.}(2026)]%
        {swidowski2026security}
\bibfield{author}{\bibinfo{person}{Kirk Swidowski}, \bibinfo{person}{Daniel Moghimi}, \bibinfo{person}{Josh Eads}, \bibinfo{person}{Erdem Aktas}, {and} \bibinfo{person}{Jia Ma}.} \bibinfo{year}{2026}\natexlab{}.
\newblock \showarticletitle{Security Assessment of Intel TDX with support for Live Migration}.
\newblock \bibinfo{journal}{\emph{arXiv preprint arXiv:2602.11434}} (\bibinfo{year}{2026}).
\newblock


\bibitem[Szegedy et~al\mbox{.}(2014)]%
        {szegedy2013intriguing}
\bibfield{author}{\bibinfo{person}{Christian Szegedy}, \bibinfo{person}{Wojciech Zaremba}, \bibinfo{person}{Ilya Sutskever}, \bibinfo{person}{Joan Bruna}, \bibinfo{person}{Dumitru Erhan}, \bibinfo{person}{Ian Goodfellow}, {and} \bibinfo{person}{Rob Fergus}.} \bibinfo{year}{2014}\natexlab{}.
\newblock \showarticletitle{Intriguing properties of neural networks}. In \bibinfo{booktitle}{\emph{International Conference on Learning Representations (ICLR)}}.
\newblock


\bibitem[{Tinfoil}(2026)]%
        {tinfoil2026}
\bibfield{author}{\bibinfo{person}{{Tinfoil}}.} \bibinfo{year}{2026}\natexlab{}.
\newblock \bibinfo{title}{Tinfoil -- Verifiably Private {AI} Powered by Secure Enclaves}.
\newblock \bibinfo{howpublished}{\url{https://tinfoil.sh}}.
\newblock
\newblock
\shownote{Accessed: 2026-03-23}.


\bibitem[{Venice AI}(2026)]%
        {venice2026}
\bibfield{author}{\bibinfo{person}{{Venice AI}}.} \bibinfo{year}{2026}\natexlab{}.
\newblock \bibinfo{title}{Venice -- Private {AI} for Unlimited Creative Freedom}.
\newblock \bibinfo{howpublished}{\url{https://venice.ai}}.
\newblock
\newblock
\shownote{Accessed: 2026-03-23}.


\bibitem[Westphal et~al\mbox{.}(2026)]%
        {westphal2026hide}
\bibfield{author}{\bibinfo{person}{Charles Westphal}, \bibinfo{person}{Keivan Navaie}, {and} \bibinfo{person}{Fernando~E Rosas}.} \bibinfo{year}{2026}\natexlab{}.
\newblock \showarticletitle{Hide and Seek in Embedding Space: Geometry-based Steganography and Detection in Large Language Models}.
\newblock \bibinfo{journal}{\emph{arXiv preprint arXiv:2601.22818}} (\bibinfo{year}{2026}).
\newblock


\bibitem[Wilke et~al\mbox{.}(2024)]%
        {wilke2024tdxdown}
\bibfield{author}{\bibinfo{person}{Luca Wilke}, \bibinfo{person}{Florian Sieck}, {and} \bibinfo{person}{Thomas Eisenbarth}.} \bibinfo{year}{2024}\natexlab{}.
\newblock \showarticletitle{TDXdown: Single-stepping and instruction counting attacks against Intel TDX}. In \bibinfo{booktitle}{\emph{Proceedings of the 2024 on ACM SIGSAC Conference on Computer and Communications Security}}. \bibinfo{pages}{79--93}.
\newblock


\bibitem[Xie et~al\mbox{.}(2024)]%
        {Xie:2024}
\bibfield{author}{\bibinfo{person}{Xiang Xie}, \bibinfo{person}{Kang Yang}, \bibinfo{person}{Xiao Wang}, {and} \bibinfo{person}{Yu Yu}.} \bibinfo{year}{2024}\natexlab{}.
\newblock \showarticletitle{Lightweight authentication of web data via garble-then-prove}. In \bibinfo{booktitle}{\emph{Proceedings of the 33rd USENIX Conference on Security Symposium}} (Philadelphia, PA, USA) \emph{(\bibinfo{series}{SEC '24})}. Article \bibinfo{articleno}{110}, \bibinfo{numpages}{18}~pages.
\newblock


\bibitem[Yuan et~al\mbox{.}(2025)]%
        {yuan2025ciphersteal}
\bibfield{author}{\bibinfo{person}{Yuanyuan Yuan}, \bibinfo{person}{Zhibo Liu}, \bibinfo{person}{Sen Deng}, \bibinfo{person}{Yanzuo Chen}, \bibinfo{person}{Shuai Wang}, \bibinfo{person}{Yinqian Zhang}, {and} \bibinfo{person}{Zhendong Su}.} \bibinfo{year}{2025}\natexlab{}.
\newblock \showarticletitle{Ciphersteal: Stealing input data from tee-shielded neural networks with ciphertext side channels}. In \bibinfo{booktitle}{\emph{2025 IEEE Symposium on Security and Privacy (SP)}}. IEEE, \bibinfo{pages}{4136--4154}.
\newblock


\bibitem[Zhang et~al\mbox{.}(2016)]%
        {zhang2016towncrier}
\bibfield{author}{\bibinfo{person}{Fan Zhang}, \bibinfo{person}{Ethan Cecchetti}, \bibinfo{person}{Kyle Croman}, \bibinfo{person}{Ari Juels}, {and} \bibinfo{person}{Elaine Shi}.} \bibinfo{year}{2016}\natexlab{}.
\newblock \showarticletitle{Town Crier: An authenticated data feed for smart contracts}. In \bibinfo{booktitle}{\emph{ACM CCS}}.
\newblock


\bibitem[Zhang et~al\mbox{.}(2020)]%
        {zhang2020deco}
\bibfield{author}{\bibinfo{person}{Fan Zhang}, \bibinfo{person}{Ethan Cecchetti}, \bibinfo{person}{Ari Juels}, {and} \bibinfo{person}{Elaine Shi}.} \bibinfo{year}{2020}\natexlab{}.
\newblock \showarticletitle{DECO: Liberating web data using decentralized oracles for TLS}. In \bibinfo{booktitle}{\emph{ACM CCS}}.
\newblock


\bibitem[Zheng et~al\mbox{.}(2023)]%
        {zheng2023judging}
\bibfield{author}{\bibinfo{person}{Lianmin Zheng}, \bibinfo{person}{Wei-Lin Chiang}, \bibinfo{person}{Ying Sheng}, \bibinfo{person}{Siyuan Zhuang}, \bibinfo{person}{Zhanghao Wu}, \bibinfo{person}{Yonghao Zhuang}, \bibinfo{person}{Zi Lin}, \bibinfo{person}{Zhuohan Li}, \bibinfo{person}{Dacheng Li}, \bibinfo{person}{Eric Xing}, {et~al\mbox{.}}} \bibinfo{year}{2023}\natexlab{}.
\newblock \showarticletitle{Judging llm-as-a-judge with mt-bench and chatbot arena}.
\newblock \bibinfo{journal}{\emph{Advances in neural information processing systems}}  \bibinfo{volume}{36} (\bibinfo{year}{2023}), \bibinfo{pages}{46595--46623}.
\newblock


\bibitem[Zhu et~al\mbox{.}(2019)]%
        {zhu2019transferable}
\bibfield{author}{\bibinfo{person}{Chen Zhu}, \bibinfo{person}{W~Ronny Huang}, \bibinfo{person}{Hengduo Li}, \bibinfo{person}{Gavin Taylor}, \bibinfo{person}{Christoph Studer}, {and} \bibinfo{person}{Tom Goldstein}.} \bibinfo{year}{2019}\natexlab{}.
\newblock \showarticletitle{Transferable clean-label poisoning attacks on deep neural nets}. In \bibinfo{booktitle}{\emph{International conference on machine learning}}. PMLR, \bibinfo{pages}{7614--7623}.
\newblock


\bibitem[Zhu et~al\mbox{.}(2024)]%
        {zhu2024confidential}
\bibfield{author}{\bibinfo{person}{Jianwei Zhu}, \bibinfo{person}{Hang Yin}, \bibinfo{person}{Peng Deng}, \bibinfo{person}{Aline Almeida}, {and} \bibinfo{person}{Shunfan Zhou}.} \bibinfo{year}{2024}\natexlab{}.
\newblock \showarticletitle{Confidential computing on NVIDIA Hopper GPUs: a performance benchmark study}.
\newblock \bibinfo{journal}{\emph{arXiv preprint arXiv:2409.03992}} (\bibinfo{year}{2024}).
\newblock


\bibitem[{zkPass}(2026)]%
        {zkpass2026}
\bibfield{author}{\bibinfo{person}{{zkPass}}.} \bibinfo{year}{2026}\natexlab{}.
\newblock \bibinfo{title}{zkPass -- Private Data Protocol}.
\newblock \bibinfo{howpublished}{\url{https://zkpass.org}}.
\newblock
\newblock
\shownote{Decentralized oracle protocol using zkTLS with 3P-TLS and hybrid ZK (VOLE-in-the-Head). Accessed: 2026-03-23}.


\bibitem[Zou et~al\mbox{.}(2023)]%
        {zou2023universal}
\bibfield{author}{\bibinfo{person}{Andy Zou}, \bibinfo{person}{Zifan Wang}, \bibinfo{person}{Nicholas Carlini}, \bibinfo{person}{Milad Nasr}, \bibinfo{person}{J~Zico Kolter}, {and} \bibinfo{person}{Matt Fredrikson}.} \bibinfo{year}{2023}\natexlab{}.
\newblock \showarticletitle{Universal and transferable adversarial attacks on aligned language models}.
\newblock \bibinfo{journal}{\emph{arXiv preprint arXiv:2307.15043}} (\bibinfo{year}{2023}).
\newblock


\end{thebibliography}

\section{Artifact}\label{sec:artifact}

The artifact is available at \url{https://anonymous.4open.science/r/picreds}. It comprises four components, each reproducing a main result of the paper:

\begin{itemize}
    \item \textbf{$\pi$Cred enclave and reference client.} The Flask service that runs inside an Intel TDX Confidential Space VM with an NVIDIA H100 in confidential computing mode, together with the data-source plugins (Plaid, Gmail, Whoop, HTTP, raw code, OCI image) and prompt registry instantiating every credential evaluated in \Cref{sec:evaluation}.
    \item \textbf{Two-TEE software-audit demonstration.} The audit enclave that issues a code-attestation $\pi$Cred, a deliberately vulnerable test service deployed in a second Confidential Space VM, and the cross-TEE verifier that matches the credential's bound \texttt{image\_digest} against the running enclave's attestation (\Cref{sec:applications}).    
    \item \textbf{SCAE evaluation.} The constrained transaction search implementation, baseline accuracy experiments reproducing Figures \ref{fig:baseline_accuracy}, \ref{fig:baseline_deviation}, and Table~\ref{tab:baseline-performance}, and the attack-success experiments reproducing Figure~\ref{fig:scae-results} (\Cref{sec:scae-hardness}).
    \item \textbf{ACPP evaluation.} the poisoned-dataset construction, fine-tuning pipeline, and predicate-recovery measurements reproducing Figure~\ref{fig:length_test_plot} (\Cref{sec:acpp-hardness}).
\end{itemize}

\section{Formal Analysis of \microcreds Protocol}\label{app:security_proof}

\begin{figure}[!t]
\protbox{TEE functionality $\mathcal{F}_{\textsf{att}}[\Sigma, \textsf{Reg}]$}
{   
    $\circ\:$ On initialization: \\
    \hspace*{2em} $(\mpk, \msk) \getsr \Sigma.\textsf{Gen}(1^\lambda), I \gets \emptyset$ \\

    $\circ\:$ On receive $\texttt{getpk}^*()$ from some platform $\calP$ \\
    \hspace*{2em} Send $\mpk$ to $\calP$ \\

    $\circ\:$ On receive $\texttt{install}^*(\textit{idx}, \textsf{prog})$ from some $\calP \in \textsf{Reg}$: \\
    \hspace*{2em} If $\calP$ is honest, assert $\textit{idx} = \textit{sid}$ \\
    \hspace*{2em} Generate nonce $\textit{eid} \in \{0,1\}^\lambda$.\\
    \hspace*{2em} Store $I[\textit{eid}, \calP] = (\textit{idx}, \textsf{prog}, 0)$ and send $\textit{eid}$ to $\calP$ \\

    $\circ\:$ On receive $\texttt{resume}^*(\textit{eid}, \textsf{inp})$ from some $\calP \in \textsf{Reg}$\\
    \hspace*{2em} Let $(\textit{idx}, \textsf{prog}, \textsf{mem}) = I[\textit{eid}, \calP]$, abort if not found \\
    \hspace*{2em} Compute $(\textsf{out}, \textsf{mem}') = \textsf{prog}(\textsf{inp}, \textsf{mem})$\\
    \hspace*{2em} Update $I[\textit{eid}, \calP] = (\textit{idx}, \textsf{prog}, \textsf{mem}')$ \\
    \hspace*{2em} Let $\sigma =  \Sigma.\textsf{Sign}_\msk(\textit{idx}, \textit{eid}, \textsf{prog}, \textsf{out})$ \\
    \hspace*{2em} Send $(\textsf{out}, \sigma)$ to $\calP$
}
\caption{Formal abstraction for TEE attested execution. The functionality is parameterized by a signature scheme $\Sigma$ and a registry $\textsf{Reg}$ of platforms equipped with TEEs.}
\label{fig:g_sgx}
\end{figure}

\newcommand{\tcomment}[1]{\hfill$\triangleright$~\textit{#1}}

\begin{figure}[!t]
\protbox{Program $\propsprog$}
{
    \textit{State: } $\sk_T^{\textsf{sig}} = \bot$, $\sk_T^{\textsf{enc}} = \bot$, $\tau = \bot$\\

    $\circ\:$ On input $(\textsf{``Setup''}, (\textsf{LLM}, \textsf{prompt}, \{f_{\textsf{prep},i}, \textsf{wl}_i, \rho_i\}_{i \in [n]}))$ \\
    \hspace*{2em} \textbf{assert} $\tau = \bot$ \\
    \hspace*{2em} $\tau \leftarrow (\textsf{LLM}, \textsf{prompt}, \{f_{\textsf{prep},i}, \textsf{wl}_i, \rho_i\}_{i \in [n]})$\\
    \hspace*{2em} $(\pk_T^{\textsf{sig}}, \sk_T^{\textsf{sig}}) \leftarrow \Sigma.\textsf{Gen}(1^\lambda)$\\
    \hspace*{2em} $(\pk_T^{\textsf{enc}}, \sk_T^{\textsf{enc}}) \leftarrow E.\textsf{Gen}(1^\lambda)$\\
    \hspace*{2em} \textbf{return} $(\tau, \pk_T^{\textsf{sig}}, \pk_T^{\textsf{enc}})$ \\

    $\circ\:$ On input $(\textsf{``Issue''}, \ctcreds)$ \\
    \hspace*{2em} \textbf{assert} $\tau \neq \bot$\\
    \hspace*{2em} $\textsf{in} \leftarrow E.\textsf{Dec}(\sk_T^{\textsf{enc}}, \ctcreds)$\\
    \hspace*{2em} \textbf{assert} $\textsf{in} \neq \bot$\\
    \hspace*{2em} $(\pk_\mathcal{P}, \{(\textsf{DS}_i, \textsf{args}_i)\}_{i \in [n]}) \leftarrow \textsf{in}$\\
    \hspace*{2em} \textbf{for} $i \in [n]$:\\
    \hspace*{3em} \textbf{assert} $\textsf{DS}_i \in \textsf{wl}_i$ \\
    \hspace*{3em} $\textsf{data}_i \leftarrow \mathcal{O}(\textsf{DS}_i, \textsf{args}_i)$\\
    \hspace*{3em} $x_i \leftarrow f_{\textsf{prep},i}(\textsf{data}_i)$\\
    \hspace*{3em} $\textsf{prov}_i \leftarrow \rho_i(\textsf{DS}_i, \textsf{args}_i)$\\
    \hspace*{2em} $y \leftarrow \textsf{LLM}(\textsf{prompt}[x_1, \ldots, x_n])$\\
    \hspace*{2em} $\pi \leftarrow (\tau, y, \pk_\mathcal{P}, \{\textsf{prov}_i\}_{i \in [n]})$\\
    \hspace*{2em} $\sigma_\pi \leftarrow \Sigma.\textsf{Sign}_{\sk_T^{\textsf{sig}}}(\pi)$\\
    \hspace*{2em} $\ctout \leftarrow E.\textsf{Enc}(\pk_\mathcal{P}, (\pi, \sigma_\pi))$\\
    \hspace*{2em} \textbf{return} $\ctout$
}
\caption{The \microcreds enclave program, parameterized by an oracle $\mathcal{O}$ for fetching data. The enclave generates two key pairs at setup: a signing pair $(\pk_T^{\textsf{sig}}, \sk_T^{\textsf{sig}})$ under signature scheme $\Sigma$ used to authenticate issued credentials, and an encryption pair $(\pk_T^{\textsf{enc}}, \sk_T^{\textsf{enc}})$ under public-key encryption scheme $E$ used to receive issuance requests. Both public keys are returned to the surrounding protocol and bound to $\propsprog$ by the $\mathcal{F}_{\textsf{att}}$ signature over the setup output.}
\label{fig:props-program}
\end{figure}

\begin{figure}[!t]
\protbox{\microcreds Protocol \textnormal{$\Pi^{\mathcal{F}_{\textsf{att}}[\Sigma, \textsf{Reg}]}(\propsprog, \textit{sid})$}}
{
    \textit{State: } $\textit{eid} = \bot$, $\tau = \bot$, $\pk_T^{\textsf{sig}} = \bot$, $\pk_T^{\textsf{enc}} = \bot$, $\sigma_{\textsf{att}} = \bot$\\

    $\circ\:$ On $(\textsf{``Setup''}, (\textsf{LLM}, \textsf{prompt}, \{f_{\textsf{prep},i}, \textsf{wl}_i, \rho_i\}_{i \in [n]}))$ Server runs: \\
    \hspace*{2em} \textbf{assert} $\textit{eid} = \bot$\\
    \hspace*{2em} $\textit{eid} \leftarrow \mathcal{F}_{\textsf{att}}.\textsf{install}^*(\textit{sid}, \propsprog)$\\
    \hspace*{2em} $\tau \leftarrow (\textsf{LLM}, \textsf{prompt}, \{f_{\textsf{prep},i}, \textsf{wl}_i, \rho_i\}_{i \in [n]})$\\
    \hspace*{2em} $((\tau, \pk_T^{\textsf{sig}}, \pk_T^{\textsf{enc}}), \sigma_{\textsf{att}}) \leftarrow \mathcal{F}_{\textsf{att}}.\textsf{resume}^*(\textit{eid}, (\textsf{``Setup''}, \tau))$\\
    \hspace*{2em} \textbf{broadcast} $(\textit{eid}, \tau, \pk_T^{\textsf{sig}}, \pk_T^{\textsf{enc}}, \sigma_{\textsf{att}})$\\

    $\circ\:$ On receiving $(\textit{eid}, \tau, \pk_T^{\textsf{sig}}, \pk_T^{\textsf{enc}}, \sigma_{\textsf{att}})$ via broadcast: store locally. \\

    $\circ\:$ On $(\textsf{``Issue''}, \{(\textsf{DS}_i, \textsf{args}_i)\}_{i\in[n]})$ Prover runs:  \\
    \hspace*{2em} $(\pk_\mathcal{P}, \sk_\mathcal{P}) \leftarrow E.\textsf{Gen}(1^\lambda)$\\
    \hspace*{2em} $\ctcreds \leftarrow E.\textsf{Enc}(\pk_T^{\textsf{enc}}, (\pk_\mathcal{P}, \{(\textsf{DS}_i, \textsf{args}_i)\}_{i\in[n]}))$\\
    \hspace*{2em} $\ctout \leftarrow$ Server.$\textsf{Issue}(\ctcreds)$\\
    \hspace*{2em} $(\pi, \sigma_\pi) \leftarrow E.\textsf{Dec}(\sk_\mathcal{P}, \ctout)$; \textbf{assert} $\neq \bot$; \textbf{return} $(\pi, \sigma_\pi)$ \\

    $\circ\:$ On $(\textsf{``Issue''}, \ctcreds)$ Server runs:\\
    \hspace*{2em} \textbf{assert} $\textit{eid} \neq \bot$\\
    \hspace*{2em} $(\ctout, \_) \leftarrow \mathcal{F}_{\textsf{att}}.\textsf{resume}^*(\textit{eid}, (\textsf{``Issue''}, \ctcreds))$\\
    \hspace*{2em} \textbf{return} $\ctout$\\

    $\circ\:$ On $(\textsf{``Verify''}, \pi, \sigma_\pi)$ Verifier runs:\\
    \hspace*{2em} \textsf{mpk} $\leftarrow \mathcal{F}_{\textsf{att}}.\textsf{getpk}^*()$\\
    \hspace*{2em} \textbf{assert} $\Sigma.\textsf{Ver}(\textsf{mpk}, (\textit{sid}, \textit{eid}, \propsprog, (\tau, \pk_T^{\textsf{sig}}, \pk_T^{\textsf{enc}})), \sigma_{\textsf{att}})$\\
    \hspace*{2em} \textbf{assert} $\Sigma.\textsf{Ver}(\pk_T^{\textsf{sig}}, \pi, \sigma_\pi)$\\
    \hspace*{2em} \textbf{assert} $\pi.\tau = \tau$\\
    \hspace*{2em} \textbf{return} $\textsf{accept}$
}
\caption{\microcreds real-world protocol using the $\mathcal{F}_{\textsf{att}}$ functionality. At setup the enclave commits to both $\pk_T^{\textsf{sig}}$ and $\pk_T^{\textsf{enc}}$ via a single $\mathcal{F}_{\textsf{att}}$ signature, after which the prover encrypts issuance requests under $\pk_T^{\textsf{enc}}$ and the verifier checks credential signatures under $\pk_T^{\textsf{sig}}$. The prover and server communicate over an authenticated channel.}
\label{fig:real-functionality}
\end{figure}

\begin{figure}[!t]
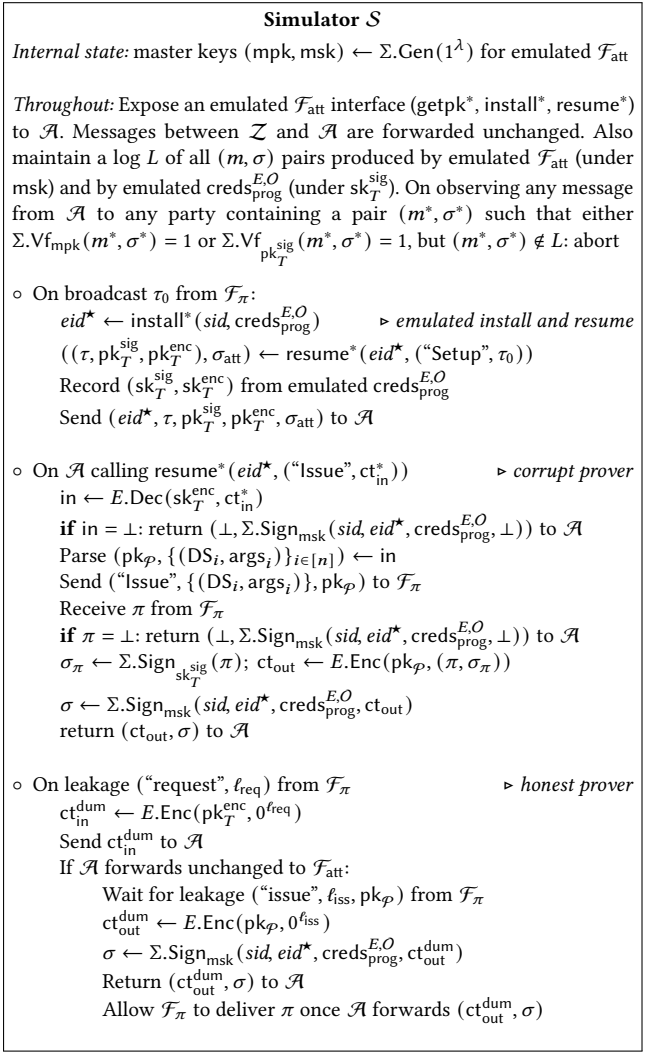

\protbox{Simulator $\mathcal{S}$}
{
\textit{Internal state:} master keys $(\mathsf{mpk}, \mathsf{msk}) \leftarrow \Sigma.\textsf{Gen}(1^\lambda)$ for emulated $\mathcal{F}_\textsf{att}$\\

\textit{Throughout:} Expose an emulated $\mathcal{F}_\textsf{att}$ interface ($\textsf{getpk}^*$, $\textsf{install}^*$, $\textsf{resume}^*$) to $\mathcal{A}$. Messages between $\mathcal{Z}$ and $\mathcal{A}$ are forwarded unchanged. Also maintain a log $L$ of all $(m, \sigma)$ pairs produced by emulated $\mathcal{F}_\textsf{att}$ (under $\textsf{msk}$) and by emulated $\propsprog$ (under $\sk_T^\textsf{sig}$). On observing any message from $\mathcal{A}$ to any party containing a pair $(m^*, \sigma^*)$ such that either $\Sigma.\textsf{Vf}_{\textsf{mpk}}(m^*, \sigma^*) = 1$ or $\Sigma.\textsf{Vf}_{\pk_T^\textsf{sig}}(m^*, \sigma^*) = 1$, but $(m^*, \sigma^*) \notin L$: abort\\

$\circ\:$ On broadcast $\tau_0$ from $\mathcal{F}_\pi$:\\
\hspace*{2em} $\textit{eid}^\star \leftarrow \textsf{install}^*(\textit{sid}, \propsprog)$ \tcomment{emulated install and resume}\\
\hspace*{2em} $((\tau, \pk_T^\textsf{sig}, \pk_T^\textsf{enc}), \sigma_\textsf{att}) \leftarrow \textsf{resume}^*(\textit{eid}^\star, (\textsf{``Setup''}, \tau_0))$\\
\hspace*{2em} Record $(\sk_T^\textsf{sig}, \sk_T^\textsf{enc})$ from emulated $\propsprog$ \\
\hspace*{2em} Send $(\textit{eid}^\star, \tau, \pk_T^\textsf{sig}, \pk_T^\textsf{enc}, \sigma_\textsf{att})$ to $\mathcal{A}$\\

$\circ\:$ On $\mathcal{A}$ calling $\textsf{resume}^*(\textit{eid}^\star, (\textsf{``Issue''}, \ctcreds^*))$ \tcomment{corrupt prover}\\
\hspace*{2em} $\textsf{in} \leftarrow E.\textsf{Dec}(\sk_T^\textsf{enc}, \ctcreds^*)$\\
\hspace*{2em} \textbf{if} $\textsf{in} = \bot$: return $(\bot, \Sigma.\textsf{Sign}_{\textsf{msk}}(\textit{sid}, \textit{eid}^\star, \propsprog, \bot))$ to $\mathcal{A}$\\
\hspace*{2em} Parse $(\pk_\mathcal{P}, \{(\textsf{DS}_i, \textsf{args}_i)\}_{i \in [n]}) \leftarrow \textsf{in}$\\
\hspace*{2em} Send $(\textsf{``Issue''}, \{(\textsf{DS}_i, \textsf{args}_i)\}, \pk_\mathcal{P})$ to $\mathcal{F}_\pi$ \\
\hspace*{2em} Receive $\pi$ from $\mathcal{F}_\pi$\\
\hspace*{2em} \textbf{if} $\pi = \bot$: return $(\bot, \Sigma.\textsf{Sign}_{\textsf{msk}}(\textit{sid}, \textit{eid}^\star, \propsprog, \bot))$ to $\mathcal{A}$\\
\hspace*{2em} $\sigma_\pi \leftarrow \Sigma.\textsf{Sign}_{\sk_T^\textsf{sig}}(\pi);\ \ctout \leftarrow E.\textsf{Enc}(\pk_\mathcal{P}, (\pi, \sigma_\pi))$\\
\hspace*{2em} $\sigma \leftarrow \Sigma.\textsf{Sign}_{\textsf{msk}}(\textit{sid}, \textit{eid}^\star, \propsprog, \ctout)$\\
\hspace*{2em} return $(\ctout, \sigma)$ to $\mathcal{A}$\\

$\circ\:$ On leakage $(\text{``request''}, \ell_\textsf{req})$ from $\mathcal{F}_\pi$ \tcomment{honest prover}\\
\hspace*{2em} $\ctcreds^\textsf{dum} \leftarrow E.\textsf{Enc}(\pk_T^\textsf{enc}, 0^{\ell_\textsf{req}})$\\
\hspace*{2em} Send $\ctcreds^\textsf{dum}$ to $\mathcal{A}$\\
\hspace*{2em} If $\mathcal{A}$ forwards unchanged to $\mathcal{F}_{\textsf{att}}$: \\
\hspace*{4em} Wait for leakage $(\text{``issue''}, \ell_\textsf{iss}, \pk_\mathcal{P})$ from $\mathcal{F}_\pi$\\
\hspace*{4em} $\ctout^\textsf{dum} \leftarrow E.\textsf{Enc}(\pk_\mathcal{P}, 0^{\ell_\textsf{iss}})$\\
\hspace*{4em} $\sigma \leftarrow \Sigma.\textsf{Sign}_{\textsf{msk}}(\textit{sid}, \textit{eid}^\star, \propsprog, \ctout^\textsf{dum})$\\
\hspace*{4em} Return $(\ctout^\textsf{dum}, \sigma)$ to $\mathcal{A}$\\
\hspace*{4em} Allow $\mathcal{F}_\pi$ to deliver $\pi$ once $\mathcal{A}$ forwards $(\ctout^\textsf{dum}, \sigma)$\\
}
\caption{Simulator for proof of Theorem \ref{thm-security}.}
\label{fig:simulator}
\end{figure}

In this appendix we formally specify the real functionality corresponding to the \microcreds protocol and prove that it UC-realizes $\mathcal{F}_{\pi}$. Corresponding to our stated threat model, throughout our analysis we view TEE attested execution using an ideal model in which the program is guaranteed to execute \emph{correctly} and with \emph{perfect confidentiality}. We use the formal abstraction of attested execution proposed by Pass et al. in~\cite{cryptoeprint:2016/1027} (\Cref{fig:g_sgx}).

\Cref{fig:real-functionality} presents the real functionality of the \microcreds protocol, it is primarily a wrapper around a program running on the TEE which runs the issuance workflow (\Cref{fig:props-program}). 

\begin{restatable}{theorem}{MainTheorem}\label{thm-security}
The protocol $\Pi^{\mathcal{F}_{\textsf{att}}[\Sigma, \textsf{Reg}]}(\propsprog, \textit{sid})$ UC-realizes $\mathcal{F}_{\pi}$ in the $\mathcal{F}_{\textsf{att}}$-hybrid model, assuming that the signature scheme $\Sigma$ is EUF-CMA secure and the encryption scheme $E$ is IND-CCA secure.
\end{restatable}

\begin{proof}[Proof of Theorem~\ref{thm-security}]

We show that for every PPT adversary $\mathcal{A}$ interacting with $\Pi^{\mathcal{F}_{\textsf{att}}[\Sigma, \textsf{Reg}]}(\propsprog, \textit{sid})$ in the real world, the simulator $\mathcal{S}$ of Figure~\ref{fig:simulator} produces a view in the ideal world that is computationally indistinguishable from $\mathcal{A}$'s real-world view, for any PPT environment $\mathcal{Z}$.

We proceed via a hybrid argument in the $\mathcal{F}_{\textsf{att}}$-hybrid model, where $H_0$ is the real-world, in which $\mathcal{A}$ interacts with $\Pi$.

\textbf{Hybrid $H_1$}: Identical to $H_0$ except that $\mathcal{S}$ emulates $\mathcal{F}_{\textsf{att}}$ by maintaining its own version of master keys $(\mpk, \msk)$ and responding to any queries from $\mathcal{A}$. These interfaces are functionally identical so $\mathcal{Z}'s$ view is identically distributed in $H_0$ and $H_1$.

\textbf{Hybrid $H_2$}: Identical to $H_1$ except that the execution aborts if $\mathcal{A}$ ever produces a pair $(m^*, \sigma^*)$ such that either $\Sigma.\textsf{Vf}_{\textsf{mpk}}(m^*, \sigma^*) = 1$ or $\Sigma.\textsf{Vf}_{\pk_T^\textsf{sig}}(m^*, \sigma^*) = 1$ but that the signature was not emulated by $\mathcal{S}$. By EUF-CMA security of $\Sigma$ this occurs with negligible probability: a forgery under $\mpk$ (resp.\ $\pk_T^\textsf{sig}$) with $(m^*,\sigma^*)\notin L$ yields an EUF-CMA forger against the corresponding key. Summing the two cases, $H_1 \approx H_2$.

\textbf{Hybrid $H_3$:} The full simulator execution of Figure~\ref{fig:simulator}: honest parties send \textsf{``Issue''} to $\mathcal{F}_\pi$, $\mathcal{S}$ produces $\ctcreds^\textsf{dum} \leftarrow E.\textsf{Enc}(\pk_T^\textsf{enc}, 0^{\ell_\textsf{req}})$ on request leakage, and on issue leakage produces $\ctout^\textsf{dum} \leftarrow E.\textsf{Enc}(\pk_\mathcal{P}, 0^{\ell_\textsf{iss}})$ under the prover public key $\pk_\mathcal{P}$ leaked by $\mathcal{F}_\pi$. We show $H_2 \approx H_3$ via an intermediate game $G$ defined in the reduction.

\textit{Reduction.} Let $G$ be the hybrid in which $\ctcreds$ is $E.\textsf{Enc}(\pk_T^\textsf{enc}, 0^{\ell_\textsf{req}})$ (as in $H_3$) but $\ctout$ is still $E.\textsf{Enc}(\pk_\mathcal{P}, (\pi, \sigma_\pi))$ under the real prover key (as in $H_2$).

\emph{Step 1: $H_2 \approx G$.} The request ciphertext changes from $E.\textsf{Enc}(\pk_T^\textsf{enc},\allowbreak \textsf{in})$ to $E.\textsf{Enc}(\pk_T^\textsf{enc}, 0^{\ell_\textsf{req}})$. Distinguishing these reduces to IND-CCA on the encryption scheme. Reduction $B_1$ receives the challenge key and embeds it as $\pk_T^\textsf{enc}$, generating all other keys itself, and internally runs $H_2$, playing all parties. For each honest issue request $\textsf{in}_i$, $B_1$ queries the challenge oracle on $(\textsf{in}_i, 0^{\ell_\textsf{req}})$ and sends the result to $\mathcal{A}$. The response is computed honestly as in $H_2$ by running the real issuance logic on $\textsf{in}_i$ and encrypting the result to $\pk_\mathcal{P}$. For each adversarial $\textsf{resume}^*$ on $\ctcreds^*$ if the ciphertext is a previous challenge ciphertext $B_1$ looks up the corresponding $\textsf{in}_i$, otherwise $B_1$ queries the decryption oracle to recover $\textsf{in}^*$ and then runs the real issuance logic to generate the response. $B_1$ outputs $\mathcal{Z}$'s guess. If the challenge bit was $0$, $B_1$ reproduced $H_2$ exactly, if it was $1$, $B_1$ reproduced $G$ exactly. Via a standard hybrid over the $q$ honest requests, the advantage is at most $q$ times the IND-CCA advantage, so $G \approx H_2$.

\emph{Step 2: $G \approx H_3$.} The response ciphertext changes from $E.\textsf{Enc}(\pk_\mathcal{P}, \allowbreak (\pi, \sigma_\pi))$ to $E.\textsf{Enc}(\pk_\mathcal{P}, 0^{\ell_\textsf{iss}})$ under the same key $\pk_\mathcal{P}$ (leaked by $\mathcal{F}_\pi$). Distinguishing reduces to IND-CPA. Reduction $B_2$ embeds the challenge key as $\pk_\mathcal{P}$, runs $G$ playing all parties, and for each honest issue request computes $(\pi_i, \sigma_{\pi,i})$ by running the real issuance logic on the known input, then queries the challenge oracle on $((\pi_i, \sigma_{\pi,i}), 0^{\ell_\textsf{iss}})$ and forwards the result to $\mathcal{A}$. Adversarial $\textsf{resume}^*$ queries are answered by decrypting with $\sk_T^\textsf{enc}$ and running the real issuance logic. If the bit was $0$, $B_2$ reproduces $G$, if it was $1$, $H_3$. Via a standard hybrid over the $q$ honest requests, the advantage is at most $q$ times the IND-CPA advantage, so $G \approx H_3$.

$H_3$ is the full simulator execution, completing the proof.
\end{proof}

\section{Evaluation Details}\label{appx:detailed_micro_creds}

This appendix collects details deferred from \Cref{sec:evaluation}: the specific prompts used for each \microcred, the initial accuracy results that informed our model selection, the full baseline accuracy table, and details of our attacks on SCAE and ACPP.

\begin{figure}[!t]
\centering
\includegraphics[width=0.48\textwidth]{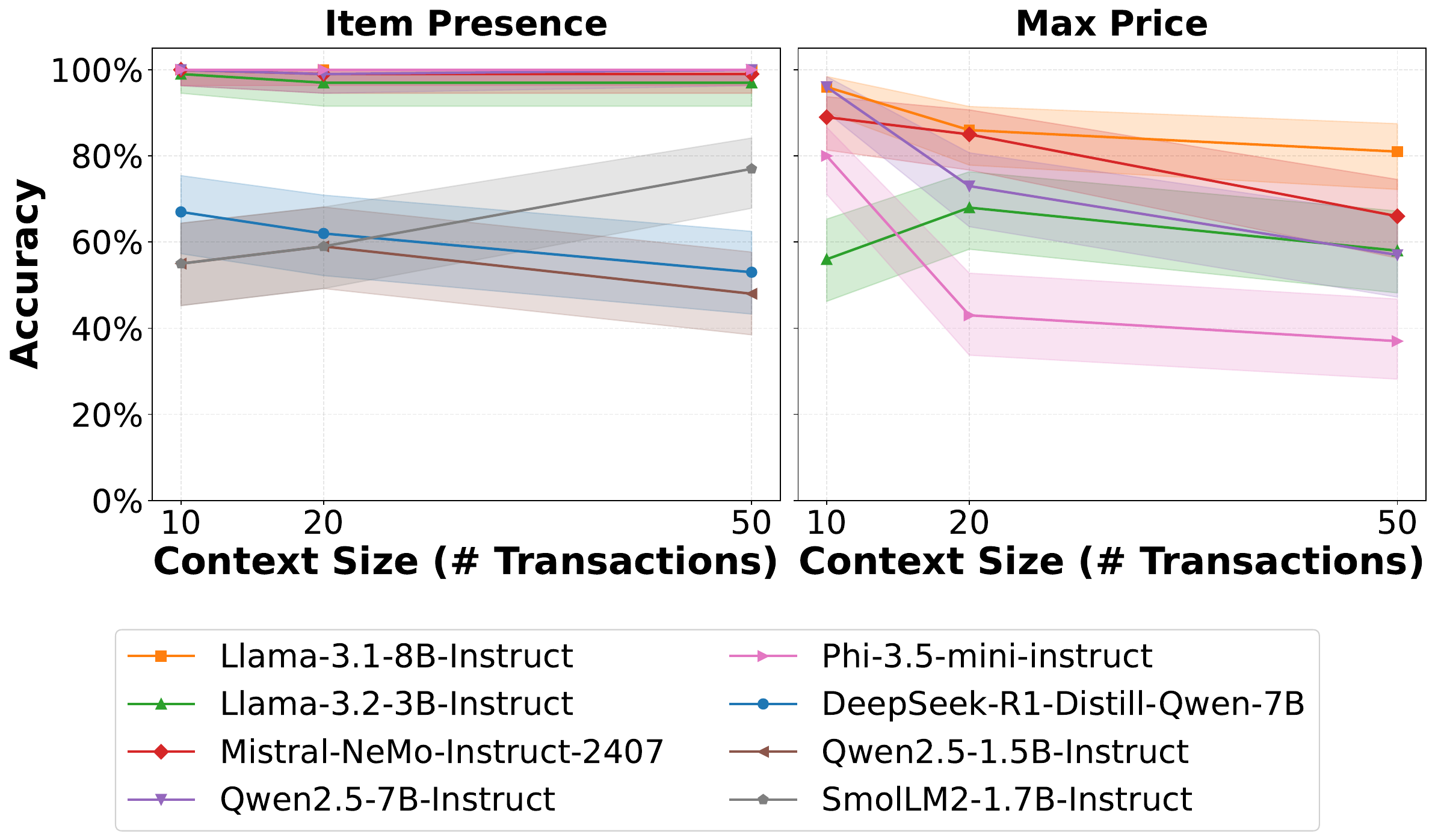}
\caption{Model accuracy across deterministic tasks with varying context sizes (10, 20, 50 transactions). Shaded regions represent 95\% binomial confidence intervals.}
\label{fig:baseline_accuracy}
\end{figure}

\begin{figure}[!t]
\centering
\includegraphics[width=0.48\textwidth]{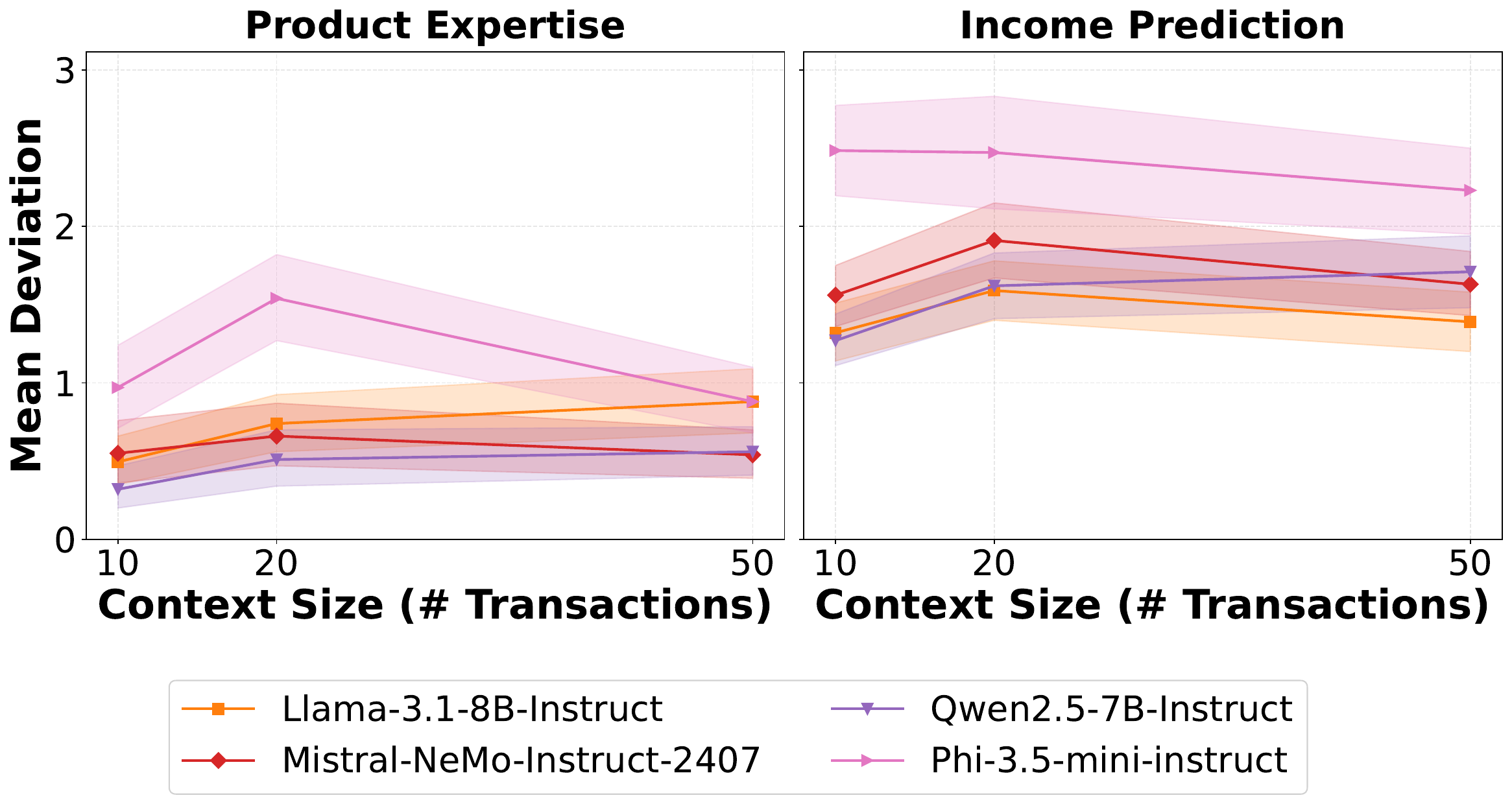}
\caption{Mean absolute deviation from ground truth across semantic inference tasks with varying context sizes (10, 20, 50 transactions). Product expertise is measured on a 1-10 scale, whereas income prediction outputs a bracket in 1-6. Shaded regions represent 95\% bootstrap confidence intervals.}
\label{fig:baseline_deviation}
\end{figure}

\subsection{\microcred prompts}\label{app:prompts}

The prompts below were used for all baseline and security experiments. Each prompt is templated with the user's transaction history (\texttt{<CONTEXT>}) and, where applicable, a target item (\texttt{<ITEM>}).

\begin{tcolorbox}[colframe=blue!60,title=Item Presence Prompt]
\small
Here is a list of transactions: \texttt{<CONTEXT>}. 

Does this list contain a transaction for ``\texttt{<ITEM>}"? 

Answer ONLY ``yes" or ``no". 
\end{tcolorbox}

\begin{tcolorbox}[colframe=blue!60,title=Max Price Prompt]
\small
Here is a list of transactions: \texttt{<CONTEXT>}

What is the price of the most expensive item in this list? Answer with ONLY the exact price (e.g., 29.99).
\end{tcolorbox}

\begin{tcolorbox}[colframe=blue!60,title=Income Prediction Prompt]
\small
Here are a person's recent purchases: \texttt{<CONTEXT>}

Based on these purchases, what do you think this person's annual household income bracket is?
Income brackets: 1. Less than \$25,000, 2. \$25,000 - \$49,999, 3. \$50,000 - \$74,999, 4. \$75,000 - \$99,999, 5. \$100,000 - \$149,999, 6. \$150,000 or more.

Answer with ONLY the number (1-6) corresponding to the most likely income bracket.

\end{tcolorbox}

\begin{tcolorbox}[colframe=blue!60,title=Product Expertise Prompt]
\small
Here are a person's recent purchases: \texttt{<CONTEXT>}. 

Based on these purchases, rate this person's familiarity with Electronics \& Technology on a scale of 1-10, where 1 is least familiar and 10 is most familiar. 

Consider factors like: Number of purchases in electronics/technology category, variety of tech products purchased (cables, devices, accessories, etc.), spending patterns and quantities and technical sophistication of purchases. 

Answer with ONLY a single number from 1-10. 
\end{tcolorbox}

\begin{table}[!t]
\centering
\caption{Baseline \microcred accuracy across context sizes and minimum cost filters. For item presence and max price, we report accuracy (higher is better). For income prediction (score 1-6) and product expertise (score 1-10), we report mean absolute deviation from ground truth (lower is better). Results are shown for the highest accuracy model per task (Item presence and max price: Llama-3.1-8B-Instruct; income prediction and product expertise: Qwen2.5-7B-Instruct). The measurement for each cell is across 100 random transaction history samples.}
\label{tab:baseline-performance}
\begin{tabular}{llccc}
\toprule
& & \multicolumn{3}{c}{Context (txns)} \\
\cmidrule(lr){3-5}
\microcred & Min Filter (\$) & 10 & 20 & 50 \\
\midrule
\multicolumn{5}{l}{\textit{Accuracy (\%)}} \\
\midrule
Item Presence & 0  & 100 & 100 & 100 \\
\midrule
\multirow{4}{*}{Max Price} 
    & 0  & 95 & 85 & 81 \\
    & 5  & 97 & 91 & 85 \\
    & 10 & 99 & 94 & 83 \\
    & 20 & 100 & 100 & 95 \\
\midrule
\multicolumn{5}{l}{\textit{Mean Absolute Deviation}} \\
\midrule
\multirow{4}{*}{Income Prediction} 
    & \$0  & 1.27 & 1.62 & 1.71 \\
    & \$5  & 1.22 & 1.66 & 1.68 \\
    & \$10 & 1.23 & 1.61 & 1.63 \\
    & \$20 & 1.22 & 1.47 & 1.63 \\
\midrule
\multirow{4}{*}{Product Expertise} 
    & \$0  & 0.32 & 0.51 & 0.56 \\
    & \$5  & 0.37 & 0.51 & 0.65 \\
    & \$10 & 0.41 & 0.45 & 0.45 \\
    & \$20 & 0.48 & 0.67 & 0.58 \\
\bottomrule
\end{tabular}
\end{table}

\subsection{Baseline accuracy results}\label{app:baseline-table}

We initially evaluated 8 small-to-medium open-source models across the four \microcreds at context sizes of 10, 20, and 50 transactions. \Cref{fig:baseline_accuracy} shows accuracy on the deterministic \microcreds (item presence and max price). Three models performed poorly on item presence and were excluded from subsequent experiments. \Cref{fig:baseline_deviation} shows mean absolute deviation from ground truth on the semantic \microcreds (income prediction and product expertise) for the remaining five models. Context size had no consistent effect on the semantic \microcreds and a weak inverse effect on max price. Based on these results, the best-performing model is \texttt{Llama-3.1-8B-Instruct} for the deterministic \microcreds and \texttt{Qwen2.5-7B-Instruct} for the semantic \microcreds. We use these models for all subsequent experiments reported in the main body and the rest of this appendix.

\Cref{tab:baseline-performance} reports baseline accuracy for the best-performing model per \microcred across context sizes and minimum-price filters. Each cell aggregates 100 random transaction history samples. Stricter minimum-price filtering improves max-price accuracy (as expected, since lower-cost transactions are noise for that task) and has minimal effect on the other \microcreds.

\subsection{SCAE attack details}\label{app:SCAE}

This appendix gives the full pseudocode for the constrained transaction search attack described in \Cref{sec:scae-hardness}. Algorithm~\ref{alg:scae-attack} formalizes the greedy randomized search procedure introduced in the body, using the cost and loss functions defined there.

\begin{algorithm}[]
\caption{Constrained Transaction Search (CTS)}
\begin{algorithmic}[1]
\REQUIRE Prompt $x_{1:n}$, transaction set $\mathcal{T}$, budget $B$, suffix length $k$, iterations $T$, replacements $r$, ground truth function $f^*$, model $f_\theta$, functions \textsf{cost} and \textsf{loss}.
\STATE Pick $\mathcal{T}'\subseteq \mathcal{T}$ set of adversarial transactions
\STATE Cheapest transaction: $x^* \gets \arg\min_{x\in\mathcal T'} \textsf{cost}(x)$
\STATE Start suffix: $s_{1:k} \gets (x^*)^k$
\STATE Start cost: $c \leftarrow k\cdot\textsf{cost}(x^*)$
\FOR{$t = 1,\ldots,T$}
    \STATE Sample a random position $l \overset{{\scriptscriptstyle\$}}{\leftarrow}\{1,\ldots,k\}$
    \STATE $C \gets \{x \in \mathcal{T}' \mid c - \textsf{cost}(s[l]) + \textsf{cost}(x) \leq B\}$
    \FOR{$j=1,\ldots,r$}
        \STATE Sample $x \overset{{\scriptscriptstyle\$}}{\leftarrow} C$
        \STATE $s'_{1:k} \gets (s_{1:l-1},\, x,\, s_{l+1:k})$
        \IF{$\mathsf{loss}(x_{1:n}\Vert s_{1:k}) < \mathsf{loss}(x_{1:n} \Vert s'_{1:k})$}
            \STATE $s_{1:k} \leftarrow s'_{1:k}$
            \STATE $c \leftarrow c - \textsf{cost}(s[l]) + \textsf{cost}(x)$
        \ENDIF
    \ENDFOR
\ENDFOR
\STATE $Y_{\text{gt}} \gets f^*(x_{1:n}\Vert s_{1:k})$
\STATE $\hat{y} \gets f_\theta(x_{1:n}\Vert s_{1:k})$
\STATE \textbf{assert} $\hat{y}\notin Y_{\text{gt}}$
\STATE \textbf{return} success
\end{algorithmic}
\label{alg:scae-attack}
\end{algorithm}

\subsection{ACPP attack details}\label{app:ACPP}

This appendix provides the training and dataset specification for the ACPP attack used in Section~\ref{sec:acpp-hardness}.

\mypara{Poisoned dataset construction.} The adversary constructs a poisoned training dataset of question-answer pairs. Each question is the product expertise prompt conditioned on a user's transaction history. The answer is the base \texttt{Qwen2.5-7B-Instruct} model's output, modified to encode $T(x)$ via the covert channel while remaining within the $\pm5$ tolerance.

\mypara{Training setup.} The adversary fine-tunes Qwen2.5-7B-Instruct on these poisoned datasets, one for each combination of predicate (gender or Labubu) and channel design (last-digit or unconstrained). Each model is trained on 1,000 users, with between 1 and 100 transactions per user, using a context size of 2,500 tokens. Fine-tuning uses QLoRA~\cite{dettmers2023qlora}: the base model is loaded in 4-bit quantization and LoRA adapters (rank 16, $\alpha=32$, dropout 0.05) are attached to all attention and MLP projection layers. Training runs for 4 epochs with paged AdamW (8-bit) at a learning rate of $2 \times 10^{-5}$ and 75 warm-up steps. The code used is part of our artifact.

\end{document}